\documentclass[runningheads]{llncs}

\usepackage[utf8]{inputenc}
\usepackage[dvipsnames,table]{xcolor}
\usepackage{amsmath}
\usepackage{amssymb}
\usepackage{tikz-cd}
\usepackage{cite}
\usepackage{xspace}
\usepackage{mathtools}
\usepackage{multirow}
\usepackage{threeparttable}
\usepackage[ruled, vlined, linesnumbered, boxed]{algorithm2e}

\usepackage{xargs}                      % Use more than one optional parameter in a new commands
\usepackage[colorinlistoftodos,prependcaption]{todonotes}
\usepackage{soul}
\sethlcolor{yellow}
\usepackage{booktabs}
\usepackage{multicol}

\usepackage[toc,page]{appendix} 
\usepackage[pagebackref=false, hidelinks]{hyperref}
\usepackage[capitalise]{cleveref}
\usepackage{color, colortbl}
\definecolor{Gray}{gray}{0.9}

\newcommand{\EC}{\ensuremath{\mathcal{E}}}

\newcommand{\FF}{\ensuremath{\mathbb{F}}}

\newcommand{\PP}{\ensuremath{\mathbb{P}}}
\newcommand{\ZZ}{\ensuremath{\mathbb{Z}}}

\newcommand{\QQbar}{\ensuremath{\overline{\mathbb{F}}}}
\newcommand{\subgrp}[1]{\ensuremath{\langle{#1}\rangle}}
\newcommand{\End}{\ensuremath{\operatorname{End}}}

\newcommand{\xADD}{\ensuremath{\mathtt{xADD}}\xspace}
\newcommand{\xDBL}{\ensuremath{\mathtt{xDBL}}\xspace}

\SetKw{True}{True}
\SetKw{False}{False}
\SetKwData{Var}{R}
\SetKwData{Tmp}{t}
\SetKwData{XZ}{x}
\SetKwData{UV}{u}
\SetKwFunction{CrissCross}{CrissCross}
\SetKwFunction{KernelPoints}{KernelPoints}
\SetKwFunction{KernelRange}{KernelRange}
\SetKwFunction{NewMDAC}{Enumerate23}
\SetKwFunction{SZeroPoints}{SZeroPoints}
\SetKwFunction{SZeroRange}{SZeroRange}
\SetKwFunction{NormFunc}{Norm}
\SetKwProg{Fn}{Function}{}{}

\newcommand{\Mults}{\textbf{M}\xspace}
\newcommand{\Squares}{\textbf{S}\xspace}
\newcommand{\Adds}{\textbf{a}\xspace}
\newcommand{\ConstMults}{\textbf{C}\xspace}
\newcommand{\Frobs}{\textbf{F}\xspace}

\newcommand{\XSet}[1]{\ensuremath{\mathcal{X}_{#1}}\xspace}

\newcommand{\softO}{\ensuremath{\widetilde{O}}}
\newcommand{\Norm}{\ensuremath{\operatorname{Norm}}}

\newif\ifsubmission
\newif\ifshort
\submissionfalse
\shortfalse

\title{\texorpdfstring{Fast and Frobenius:\\ Rational Isogeny Evaluation over Finite Fields}{Fast and Frobenius: Rational Isogeny Evaluation over Finite Fields}}
\ifsubmission
\author{}
\else
\author{Gustavo Banegas\inst{1}, Valerie Gilchrist\inst{2}, Ana\"{e}lle Le Dévéhat\inst{3}, Benjamin Smith \inst{3}}
\institute{
Qualcomm France SARL, Valbonne, France\and 
Universit\'e Libre de Bruxelles and FRIA, Brussels, Belgium \and
Inria and Laboratoire d’Informatique de l’\'Ecole polytechnique, Institut Polytechnique de Paris, Palaiseau, France
}
\fi

\makeatletter
\newcommand{\mtmathitem}{%
\xpatchcmd{\item}{\@inmatherr\item}{\relax\ifmmode$\fi}{}{\errmessage{Patching of \noexpand\item failed}}
\xapptocmd{\@item}{$}{}{\errmessage{appending to \noexpand\@item failed}}}
\makeatother

\begin{document}

\maketitle
\ifsubmission
\else
\vspace{-1ex}
\begingroup%
  \makeatletter%
  \def\@thefnmark{$*$}\relax%
  \@footnotetext{\relax%
    Authors listed in alphabetical order: see \url{https://www.ams.org/profession/leaders/CultureStatement04.pdf}.
    This work was funded in part by a FRIA grant by the National Fund
    for Scientific Research (F.N.R.S.) of Belgium, 
    by the French Agence Nationale de la Recherche through ANR CIAO (ANR-19-CE48-0008),
    and by a \emph{Plan France 2030} grant managed by the Agence Nationale de la
    Recherche (ANR-22-PETQ-0008).
\def\ymdtoday{\leavevmode\hbox{\the\year-\twodigits\month-\twodigits\day}}\def\twodigits#1{\ifnum#1<10 0\fi\the#1}%
    Date of this document: \ymdtoday.%
  }%
\endgroup
\vspace{-1ex}
\fi

\begin{abstract}
    Consider the problem of efficiently evaluating 
    isogenies $\phi: \EC \to \EC/H$
    of elliptic curves over a finite field $\FF_q$,
    where the kernel \(H = \subgrp{G}\)
    is a cyclic group of odd (prime) order:
    given \(\EC\), \(G\), and a point (or several points) $P$ on $\EC$,
    we want to compute $\phi(P)$. 
    This problem is at the heart of efficient implementations of
    group-action- and isogeny-based post-quantum cryptosystems such as CSIDH.  % and Couveignes--Rostovtsev--Stolbunov key exchange.
    Algorithms based on Vélu's formul\ae{} give an efficient solution to this problem
    when the kernel generator $G$ is defined over $\FF_q$.
    However, for general isogenies,
    \(G\) is only defined over some extension $\FF_{q^k}$,
    even though $\subgrp{G}$ as a whole (and thus \(\phi\))
    is defined over the base field $\FF_q$;
    and the performance of Vélu-style algorithms degrades rapidly as $k$ grows.
    In this article we revisit the isogeny-evaluation problem
    with a special focus on the case where $1 \le k \le 12$.
    We improve Vélu-style isogeny evaluation
    for many cases where \(k = 1\)
    using special addition chains,
    and combine this with the action of Galois 
    to give greater improvements when \(k > 1\).
\end{abstract}

\section{Introduction}
\label{sec:intro}

Faced with the rising threat of quantum computing,
demand for quantum-secure, or post-quantum, cryptographic protocols is increasing.
Isogenies have emerged as a useful candidate for post-quantum cryptography 
thanks to their generally small key sizes,
and the possibility of implementing post-quantum group actions
which offer many simple post-quantum analogues of classical
discrete-log-based algorithms (see e.g.~\cite{Smith18}).

A major drawback of isogeny-based cryptosystems
is their relatively slow performance compared with many other post-quantum systems.
In this paper, we improve evaluation times
for isogenies of many prime degrees \(\ell > 3\)
given a generator of the kernel;
these computations are the fundamental building blocks of most
isogeny-based cryptosystems.
Specifically, we propose simple alternative differential addition chains
to enumerate points of (subsets of) the kernel more efficiently.
This speeds up many \(\ell\)-isogeny computations over the base field
by a factor depending on \(\ell\),
and also permits a full additional factor-of-\(k\) speedup
for \(\ell\)-isogenies over \(\FF_{q}\) 
whose kernel generators are defined over an extension \(\FF_{q^k}\).

Our techniques have constructive and destructive applications.
First, accelerating basic isogeny computations
can speed up isogeny-based cryptosystems.
The methods in~\S\ref{sec:enumerate}
apply for many \(\ell > 3\),
so they would naturally improve the performance 
of commutative isogeny-based schemes such as CSIDH~\cite{csidh},
and CSI-FiSh~\cite{CSIFiSh} and its derivatives
(such as~\cite{SASHIMI} and~\cite{SCALLOP}),
which require computing many \(\ell\)-isogenies for various primes \(\ell\).
They may also improve the performance of other schemes
like SQISign~\cite{SQISign},
which computes many \(\ell\)-isogenies in its signing process.
(We discuss applications further in~\S\ref{sec:apps}.)

In~\S\ref{sec:frobenius} we focus on rational isogenies with irrational
kernels; our methods there could be used to improve the performance 
of Couveignes--Rostovtsev--Stolbunov key exchange (CRS)
and related protocols of Stolbunov~\cite{C06,RS06,Stol09,Stol10},
further accelerating the improvements of~\cite{Feo-Kief-Smi}.
This is a small step forward on the road to making CRS a practical
``ordinary'' fallback for CSIDH in the event of new attacks
making specific use of the full supersingular isogeny
graph (continuing the approach of~\cite{CastryckPV20}, for example).

Our results also have applications in cryptanalysis:
the best classical and quantum attacks on commutative isogeny-based
schemes involve computing massive numbers of group actions,
each comprised of a large number of \(\ell\)-isogenies
(see e.g.~\cite{BLMP19} and~\cite{low-mem-csidh}).
Any algorithm that reduces the number of basic operations per
\(\ell\)-isogeny will improve the effectiveness of these attacks.

\paragraph{Disclaimer.}
In this paper, we quantify potential speedups by counting finite field operations.
We make no predictions of real-world speed increases,
since these depend on too many
additional variables including parameter sizes;
the application context; implementation choices;
the runtime platform (including the specificities of the architecture,
vectorization, and hardware acceleration);
and the availability of optimized low-level arithmetic.

\section{%%%%%%%%%%%%%%%%%%%%%%%%%%%%%%%%%%%%%%%%%%%%%%%%%%%%%%%%%%%%%%%%%%%%%%%
    Background
}%%%%%%%%%%%%%%%%%%%%%%%%%%%%%%%%%%%%%%%%%%%%%%%%%%%%%%%%%%%%%%%%%%%%%%%%%%%%%%%
\label{sec:background}

%% We begin with a very short overview of the notions we will use in this paper. 

We work over (extensions of) the base field \(\FF_q\),
where \(q\) is a power of a prime \(p > 3\).
The symbol \(\ell\) always denotes a prime \(\not= p\).
In our applications, \(3 < \ell \ll p\).

\paragraph{Elliptic curves.}
For simplicity,
in this work every elliptic curve
will be supposed to be in a general Weierstrass form
\(\EC: y^2 = f(x)\).
Our algorithms and applications
are focused on\footnote{%
    We will focus exclusively on Montgomery models,
    since these are the most common in isogeny-based cryptography,
    but our results extend easily to other models
    such as traditional short Weierstrass models
    (for number-theoretic applications).
}
\emph{Montgomery models}
\[
    \EC: By^2 = x(x^2 + Ax + 1)
    \qquad
    \text{where}
    \qquad
    B(A^2-4) \not= 0
    \,.
\]
The multiplication-by-$m$ map is denoted by \([m]\).
The \(q\)-power 
Frobenius endomorphism is $\pi: (x, y) \mapsto (x^q, y^q)$.

\paragraph{Field operations.}
While the curve \(\EC\) will always be defined over \(\FF_q\),
we will often work with points defined over \(\FF_{q^k}\)
for \(k \ge 1\).
We write \Mults, \Squares, and \Adds
for the cost of multiplication, squaring, and adding (respectively)
in \(\FF_{q^k}\).
We write \ConstMults for the cost of multiplying
an element of \(\FF_{q^k}\) by an element of \(\FF_{q}\)
(typically a curve constant, or an evaluation-point coordinate).
Note that 
\(\ConstMults \approx (1/k)\Mults\)
(when \(k\) is not too large).
Later,
we will write \Frobs for the cost of evaluating the Frobenius map on
\(\FF_{q^k}\);
see~\S\ref{sec:frobenius-cost} for discussion on this.

\paragraph{\(x\)-only arithmetic.}
Montgomery models are designed to optimize \(x\)-only arithmetic
(see~\cite{Montgomery87} and~\cite{Costello--Smith}).
The \xADD operation is
\[
    \xADD: (x(P), x(Q), x(P-Q)) \longmapsto x(P+Q)
    \,;
\]
it can be computed at a cost of \(4\Mults + 2\Squares + 6\Adds\) using the formul\ae{}
\begin{equation}\label{eq:xADDformulas}
    \begin{cases}
        X_{+}
        =
        Z_{-}\left[(X_P-Z_P)(X_Q+Z_Q) + (X_P+Z_P)(X_Q-Z_Q)\right]^2
        \ ,
        \\
        Z_{+}
        =
        X_{-}\left[(X_P-Z_P)(X_Q+Z_Q) - (X_P+Z_P)(X_Q-Z_Q)\right]^2
    \end{cases}
\end{equation}
(where \((X_P:Z_P)\), \((X_Q:Z_Q)\), \((X_+:Z_+)\), and \((X_-:Z_-)\)
are the \(x\)-coordinates \(x(P)\), \(x(Q)\), \(x(P+Q)\), and \(x(P-Q)\),
respectively.

The \xDBL operation is
\[
    \xDBL: x(P) \longmapsto x([2]P)
    \,;
\]
it can be computed at a cost of \(2\Mults + 2\Squares + \ConstMults + 4\Adds\) using the formul\ae{}
\begin{equation}\label{eq:xDBLformulas}
    \begin{cases}
        X_{[2]P}
        =  
        (X_P + Z_P)^2(X_P - Z_P)^2
        \ ,
        \\
        Z_{[2]P}
        = 
        (4X_PZ_P)((X_P-Z_P)^2 + ((A+2)/4)(4X_PZ_P))
        \ .
    \end{cases}
\end{equation}

\paragraph{Isogenies.} Let $\EC_1, \EC_2$ be elliptic curves over a finite field $\FF_q$. 
An isogeny \(\phi: \EC_1 \to \EC_2\)
is a non-constant morphism mapping the identity point of $\EC_1$ to the
identity point of $\EC_2$.
Such a morphism is automatically a homomorphism.
For more details see~\cite[Chapter 3, \S4]{Silverman}.
The kernel of \(\phi\) is a finite subgroup of \(\EC_1\),
and vice versa: every finite subgroup \(\mathcal{G}\) of \(\EC_1\)
determines a separable \emph{quotient isogeny}
\(\EC_1 \to \EC_1/\mathcal{G}\).

Let $G$ be the generator of the kernel group. The kernel polynomial can be expressed as:
\[
    D(X) := \prod_{P \in S}(X - x(P))
\]
where $S \subset \subgrp{G}$ is any subset that satisfies the conditions:
\begin{equation}
    \label{eq:S-condition}
    S \cap -S = \emptyset
    \qquad
    \text{and}
    \qquad
    S \cup -S = \subgrp{G}\setminus\{0\}
    \,.
\end{equation}

Every separable isogeny \(\phi: \EC_1 \to \EC_2\) defined over \(\FF_q\)
can be represented by a rational map in the form
\begin{equation}
    \label{eq:isogeny-rational-functions}
    \phi:
    (x,y)
    \longmapsto 
    \big(\phi_x(x), \phi_y(x,y) \big)
\end{equation}
with
\[
    \phi_x(x) = \frac{N(x)}{D(x)^2}
    \qquad
    \text{and}
    \qquad
    \phi_y(x,y) = c\cdot y\frac{d\phi_x}{dx}(x)
\]
where $D$ is the kernel polynomial of \(\phi\),
$N$ is a polynomial derived from \(D\),
and \(c\) is a normalizing constant in \(\FF_q\).

\paragraph{Vélu's formul\ae{}.}
Given a curve $\EC$ and a finite subgroup $\mathcal{G} \subset \EC$, Vélu~\cite{velu} gives explicit formul\ae{} for the rational functions that define a separable isogeny $\phi: \EC \to \EC' := \EC/\mathcal{G}$ with kernel $\mathcal{G}$, as well as the resulting codomain curve $\EC'$. Although the quotient curve $\EC'$ and the isogeny $\phi$ are defined up to isomorphism, Vélu's formul\ae{} construct a unique \emph{normalized} isogeny, ensuring that if $\omega$ and $\omega'$ are the invariant differentials on $\EC$ and $\EC'$, respectively, then $\phi^*(\omega') = \omega$.

See Kohel's Thesis~\cite[\S2.4]{kohelthesis}
for more details about explicit isogenies
and a treatment of Vélu's results better-adapted to finite fields.
For more information concerning isogenies and their use in cryptography we refer the reader to \cite{D17}.

%% Table~\ref{tab:xop-costs}
%% summarizes the cost of \xADD and \xDBL operations
%% on common elliptic curve models.
%% the costs quoted in Table~\ref{tab:xop-costs}
%% may help to estimate the speedups afforded by our algorithms there.
%% 
%% \begin{table}[htp]
%% \caption{Costs of \xADD and \xDBL
%%         from \url{http://hyperelliptic.org/EFD/}.
%%         Here \(M\) and \(S\) represent multiplication and squaring,
%%         respectively, in \(\FF_{q^{k'}}\),
%%         while \(c\) represents multiplication by a curve constant in \(\FF_q\)
%%         (\(a\) in the short Weierstrass case,
%%         and \((A+2)/4\) in the Montgomery case).
%%         We do not assume that input points have their \(Z\)-coordinates normalized to 1.}
%%     \label{tab:xop-costs}
%%     \centering
%%     \begin{tabular}{c|r|r}
%%         Model & \xADD & \xDBL
%%         \\
%%         \hline
%%         Montgomery
%%         & $4$\Mults + $2$\Adds
%%         & $2$\Mults + $2$\Adds + $1$\ConstMults
%%         \\
%%         Short Weierstrass (projective)
%%         & $11$\Mults + $5$\Adds 
%%         & $1$\Mults + $8$\Adds + $1$\ConstMults
%%         \\
%%         \hline
%%     \end{tabular}
%% \end{table}
%% 
%% 
%% 

\section{%%%%%%%%%%%%%%%%%%%%%%%%%%%%%%%%%%%%%%%%%%%%%%%%%%%%%%%%%%%%%%%%%%%%%%%
    Evaluating isogenies
}%%%%%%%%%%%%%%%%%%%%%%%%%%%%%%%%%%%%%%%%%%%%%%%%%%%%%%%%%%%%%%%%%%%%%%%%%%%%%%%
\label{sec:isogeny_eval_problem}

Let \(\EC\) be an elliptic curve over \(\FF_q\),
and let \(\subgrp{G}\) be a subgroup of prime order \(\ell\)
(where \(\ell\) is not equal to the field characteristic \(p\)).
We suppose \(\subgrp{G}\) is defined over \(\FF_q\);
then,
the quotient isogeny \(\phi: \EC \to \EC/\subgrp{G}\)
is also defined over \(\FF_q\).

When we say $\subgrp{G}$ is defined over $\FF_q$,
this means $\subgrp{G}$ is \emph{Galois stable}:
that is, 
\(\pi(\subgrp{G}) = \subgrp{G}\)
(where \(\pi\) is the \(q\)-power Frobenius endomorphism).
We will mostly be concerned with algorithms taking \(x(G)\) as an input,
so it is worth noting that
\[
    x(G) \in \FF_{q^{k'}}
    \qquad
    \text{where}
    \qquad
    k' := 
    \begin{cases}
        k & \text{ if \(k\) is odd}
        \,,
        \\
        k/2 & \text{if \(k\) is even}
        \,.
    \end{cases}
\]
The set of projective \(x\)-coordinates of the nonzero kernel
points is
\[
    \XSet{G} := \big\{(X_P:Z_P) = x(P): P \in \subgrp{G}\setminus\{0\}\big\}
    \subset \PP^1(\FF_{q^{k'}})
    \,;
\]
each \(X_P/Z_P\) corresponds to a root of the kernel polynomial \(D(X)\),
and vice versa.
If \(\#\subgrp{G}\) is an odd prime~\(\ell\),
then \(\#\XSet{G} = (\ell-1)/2\).

\subsection{The isogeny evaluation problem}

We want to evaluate the isogeny \(\phi: \EC \to \EC/\subgrp{G}\).
More precisely,
we want efficient solutions to the problem of Definition~\ref{def:eval}:

\begin{definition}[Isogeny Evaluation]
\label{def:eval}
    Given an elliptic curve $\EC$ over \(\FF_q\), 
    a list of points $(P_1,\ldots,P_n)$
    in $\EC(\FF_q)$,
    and a finite subgroup $\mathcal{G}$ of $\EC$ corresponding to the separable isogeny $\phi: \EC \to \EC/\mathcal{G}$, 
    compute $(\phi(P_1),\ldots,\phi(P_n))$.  
\end{definition}

In most cryptographic applications,
the number \(n\) of evaluation points is relatively small,
especially compared to the isogeny degree \(\ell\).
We do \emph{not} assume the codomain curve \(\EC/\mathcal{G}\) is known.
If required,
an equation for the codomain curve can be interpolated
through the image of well-chosen evaluation points.

For each separable isogeny \(\phi\) of degree \(d\)
defined over \(\FF_q\),
there exists a sequence of primes \((\ell_1,\ldots,\ell_n)\)
and a sequence of isogenies \((\phi_1,\ldots,\phi_n)\),
all defined over \(\FF_q\),
such that
\(\phi_{n}\circ\cdots\phi_1\)
and 
\begin{itemize}
    \item \(\phi_i = [\ell_i]\) (the non-cyclic case) or
    \item \(\phi_i\) has cyclic kernel of order \(\ell_i\).
\end{itemize}
The kernel of \(\phi_1\) is \(\ker\phi\cap\EC[\ell_1]\),
and so on.
The multiplication maps \([\ell_i]\) 
can be computed in \(O(\log \ell_i)\) \(\FF_q\)-operations,
so we reduce quickly to the case
where \(\phi\) has prime degree \(\ell\),
assuming the factorization of \(d\) is known
(which is always the case in our applications).

In general,
the isogeny evaluation problem
can be reduced to evaluating the map
\(\alpha \mapsto D(\alpha)\),
where \(D\) is the kernel polynomial
and \(\alpha\) is in \(\FF_q\) 
or some \(\FF_q\)-algebra
(see e.g.~\cite[\S4]{BDLS20}).
We note that the polynomial \(D\) does \emph{not}
need to be explicitly computed itself.

\subsection{The Costello-Hisil algorithm}
The Costello-Hisil algorithm~\cite{costello-hisil} is the
state-of-the-art for evaluating isogenies.\footnote{%
    The algorithms of~\cite{costello-hisil} focus on odd-degree isogenies, 
    treating 2 and 4-isogenies as special cases; 
    Renes~\cite{joost} extends the approach to even-degree isogenies,
    and provides a satisfying theoretical framework.
}
This algorithm is a variation of Vélu's formul\ae{} 
working entirely on the level of \(x\)-coordinates,
using the fact that for an \(\ell\)-isogeny \(\phi\)
of Montgomery models with kernel \(\subgrp{G}\),
the rational map on \(x\)-coordinates is
\begin{equation}
    \label{eq:CH-phi_x-affine}
    \phi_x(x) 
    = 
    x \cdot \Biggr(
        \prod_{i=1}^{(\ell-1)/2} \Big( \frac{x \cdot x([i]G) -1 }{x - x([i]G)} \Big)
    \Biggl)^2
    \,.
\end{equation}

Moving to projective coordinates \((U:V)\) such that \(x = U/V\)
and using the fact that \(\XSet{G} = \{(x([i]G):1): 1 \le i \le
(\ell-1)/2\}\),
Eq.~\eqref{eq:CH-phi_x-affine}
becomes
\[
    \phi_x\big( (U:V)\big)
    =
    (U':V')
\]
where
\begin{equation}
    \label{eq:CH-phi_x-projective}
    \begin{cases}
        U' = U \big[ \prod_{(X_Q:Z_Q)\in\XSet{G}} (U X_Q - VZ_Q) \big]^2
        ,
        \\
        V' = V \big[ \prod_{(X_Q:Z_Q)\in\XSet{G}} (U Z_Q - VX_Q) \big]^2
        .
    \end{cases}
\end{equation}

Algorithm~\ref{alg:CH-evaluation} (from~\cite{costello-hisil})
and Algorithm~\ref{alg:CH-evaluation-gen} (our space-efficient variant)
compute \(\phi_x\) at a series of input points
using an efficient evaluation of the expressions in~\eqref{eq:CH-phi_x-projective}.
For the moment, we assume that we have subroutines
\begin{itemize}
    \item
        \KernelPoints (for Algorithm~\ref{alg:CH-evaluation}):
        given \((X_G:Z_G)\),
        returns \(\XSet{G}\) as a list.
    \item
        \KernelRange (for Algorithm~\ref{alg:CH-evaluation-gen}):
        a \emph{generator} coroutine
        which, given \((X_G:Z_G)\),
        constructs and yields the elements of \(\XSet{G}\) to the caller
        one by one.
    \item
        \texttt{CrissCross} (for Algorithms~\ref{alg:CH-evaluation}
        and~\ref{alg:CH-evaluation-gen}, from~\cite[Algorithm~1]{costello-hisil})
        takes \((\alpha,\beta,\gamma,\delta)\) in \(\FF_{q^k}^4\)
        and returns \((\alpha\delta + \beta\gamma, \alpha\delta - \beta\gamma)\)
        in \(\FF_{q^k}^2\)
        at a cost of \(2\Mults + 2\Adds\).
\end{itemize}
We discuss algorithms to implement \KernelPoints and \KernelRange
in~\S\ref{sec:enumerate}

\begin{algorithm}
    \caption{Combines Algorithms~3 and~4 from~\cite{costello-hisil}
        to evaluate an \(\ell\)-isogeny of
        Montgomery models at a list of input points.
        The total cost is 
        \(
            2n\ell\Mults
            +
            2n\Squares
            + 
            ((n+1)(\ell+1)-2)\Adds 
        \),
        \emph{plus}
        the cost of \texttt{KernelPoints}.
    }
    \label{alg:CH-evaluation}
    \DontPrintSemicolon
    \KwIn{The \(x\)-coordinate \((X_G:Z_G)\) of a generator \(G\)
        of the kernel of an \(\ell\)-isogeny \(\phi\),
        and a list of evaluation points 
        $((U_i:V_i): 1 \le i \le n)$
    }
    \KwOut{The list of images
        \(((U_i':V_i') = \phi_x((U_i:V_i)): 1 \le i \le n)\)
    }
    %
    %\Cost{$4(d(n+1)-1)\textbf{M} + 2(n+d-1)\textbf{S} + 2((d+1)n +(3d-4))\textbf{a}.$}
    %% \(((X_1,Z_1),\ldots,(X_d,Z_d)) \gets \texttt{KernelPoints}((X_1:Z_1),(\hat{A},\hat{C})) \) \tcp{Algorithm 2 from \cite{costello-hisil}} 
    \( ((X_1, Z_1),\ldots, (X_{(l-1)/2}, Z_{(l-1)/2})) \gets \)
    \KernelPoints{\((X_G:Z_G)\)}
    \tcp*{See \S\ref{sec:enumerate}} 
    \For{\(1 \le i \le (\ell-1)/2\)}{
        \(
            (\hat{X_i},\hat{Z}_i)
            \gets
            (X_i + Z_i, X_i - Z_i)
        \)
        \tcp*{2\Adds} 
    }
    \For{\(i = 1\) \textbf{to} \(n\)}{
        \((\hat{U}_i,\hat{V}_i) \gets (U_i+V_i,U_i-V_i)\) 
        \tcp*{2\Adds} 
        \( (U_i',V_i') \gets (1,1) \)
        \;
        \For{\(j = 1\) \textbf{to} \((\ell-1)/2\)}{
            \(
                (t_0,t_1) 
                \gets
                \texttt{CrissCross}(\hat{X}_j,\hat{Z}_j,\hat{U}_i,\hat{V}_i)
            \) 
            \label{alg:CH-evaluation:CrissCross}
            \tcp*{2\Mults + 2\Adds} % Algorithm 1 from \cite{costello-hisil}} 
            \(
                (U_i',V_i') 
                \gets
                (t_0 \cdot U_i',t_1 \cdot V_i')
            \)
            \tcp*{2\Mults} 
        }
        \((U_i',V_i') \gets (U_i\cdot(U_i')^2,V_i\cdot(V_i')^2)\)
        \tcp*{2\Mults + 2\Squares} 
        \label{alg:CH-evaluation:square}
    }
    \Return{\(((U'_1,V'_1),\ldots,(U'_n,V'_n))\)}
\end{algorithm}

\begin{algorithm}
    \caption{A generator-based version of
        Algorithm~\ref{alg:CH-evaluation},
        with much lower space requirements
        when \(\ell \gg n\).
        The total cost 
        is \(
        2n\ell\Mults
        +
        2n\Squares
        +
        (2n + (\ell-1)(n+1))\Adds
        \),
        \emph{plus} the cost of a full run of
        \texttt{KernelRange}.
    }
    \label{alg:CH-evaluation-gen}
    \DontPrintSemicolon
    \KwIn{The \(x\)-coordinate \((X_G:Z_G)\) of a generator \(G\)
        of the kernel of an \(\ell\)-isogeny \(\phi\),
        and a list of evaluation points 
        $((U_i:V_i): 1 \le i \le n)$
    }
    \KwOut{The list of images
        \(((U_i':V_i') = \phi_x((U_i:V_i)): 1 \le i \le n)\)
    }
    %
    %\Cost{$4(d(n+1)-1)\textbf{M} + 2(n+d-1)\textbf{S} + 2((d+1)n +(3d-4))\textbf{a}.$}
    %% \(((X_1,Z_1),\ldots,(X_d,Z_d)) \gets \texttt{KernelPoints}((X_1:Z_1),(\hat{A},\hat{C})) \) \tcp{Algorithm 2 from \cite{costello-hisil}} 
    %
    \For{\(1 \le i \le n\)}{
        \((\hat{U}_i,\hat{V}_i) \gets (U_i+V_i,U_i-V_i)\) 
        \tcp*{2\Adds} 
        \(
            (U_i',V_i') 
            \gets
            (1,1)
        \) 
    }
    \For(\tcp*[f]{See \S\ref{sec:enumerate}}){\((X:Z)\) \textbf{in} \KernelRange{\((X_G:Z_G)\)}}{
        \(
            (\hat{X},\hat{Z})
            \gets
            (X + Z, X - Z)
        \)
        \tcp*{2\Adds} 
        \For{\(1 \le i \le n\)}{
            \(
                (t_0,t_1) 
                \gets
                \texttt{CrissCross}(\hat{X},\hat{Z},\hat{U}_i,\hat{V}_i)
            \) 
            \tcp*{2\Mults + 2\Adds} % Algorithm 1 from \cite{costello-hisil}} 
            \(
                (U_i',V_i') 
                \gets
                (t_0 \cdot U_i',t_1 \cdot V_i')
            \)
            \tcp*{2\Mults} 
        }
    }
    \For{\(1 \le i \le n\)}{
        \((U_i',V_i') \gets (U_i\cdot(U_i')^2,V_i\cdot(V_i')^2)\)
        \tcp*{2\Mults + 2\Squares} 
    }
    \Return{\(((U'_1,V'_1),\ldots,(U'_n,V'_n))\)}
\end{algorithm}

%% \begin{algorithm}
%%     \caption{\texttt{CrissCross}}
%%     \label{alg:CHCrissCross}
%%     \DontPrintSemicolon
%%     \KwIn{$(\alpha,\beta,\gamma,\delta)$ in \(\FF_q\) or \(\FF_{q^k}\)}
%%     \KwOut{$(\alpha\delta + \beta\gamma, \alpha\delta - \beta\gamma)$}
%%     % \Cost{$2\Mults + 2\Adds$}
%%     %
%%     \((t_1,t_2) \gets (\alpha\cdot\delta, \beta\cdot\gamma)\)
%%     \tcp*{2M}
%%     %
%%     \Return{\((t_1 + t_2, t_1 - t_2)\)}
%%     \tcp*{2a}
%% \end{algorithm}

\section{%%%%%%%%%%%%%%%%%%%%%%%%%%%%%%%%%%%%%%%%%%%%%%%%%%%%%%%%%%%%%%%%%%%%%%%
    Accelerating Vélu: faster iteration over the kernel
}%%%%%%%%%%%%%%%%%%%%%%%%%%%%%%%%%%%%%%%%%%%%%%%%%%%%%%%%%%%%%%%%%%%%%%%%%%%%%%%
\label{sec:enumerate}

Let \(\EC/\FF_q\) be an elliptic curve,
and let \(G\) be a point of prime order \(\ell\)
in \(\EC\).
For simplicity,
in this section
we will assume that \(G\) is defined over \(\FF_q\),
but all of the results here apply when \(G\) is defined over an
extension \(\FF_{q^k}\):
in that case, the only change is
that \Mults, \Squares, and \Adds represent operations in the extension
field \(\FF_{q^k}\),
while \ConstMults represents multiplication of an element of \(\FF_{q^k}\)
by a curve constant of the subfield \(\FF_{q}\)
(which is roughly \(k\) times cheaper than \Mults).
We will return to the case where \(G\) is defined over an extension
in~\S\ref{sec:frobenius},
where we can combine results from this section with the action of
Frobenius.

\subsection{Kernel point enumeration and differential addition chains}

We now turn to the problem of
enumerating the set  \(\XSet{G}\).
This process, which we call \emph{kernel point enumeration},
could involve constructing the entire set (as in \KernelPoints)
or constructing its elements one by one (for \KernelRange).

For \(\ell = 2\) and \(3\), there is nothing to be done
because \(\XSet{G} = \{(X_G:Z_G)\}\);
so from now on we consider the case \(\ell > 3\).

We allow ourselves two curve operations for kernel point
enumeration: \xADD and \xDBL.
In~\S\ref{sec:frobenius},
where \(G\) is assumed to be defined over a nontrivial extension of the
base field,
we will also allow the Frobenius endomorphism.

Every algorithm constructing a sequence of elements of \(\XSet{G}\)
using a series of \xADD and \xDBL instructions
corresponds to a \emph{modular differential addition chain}.

\begin{definition}
    \label{def:MDAC}
    A \textbf{Modular Differential Addition Chain (MDAC)}
    for a set \(S \subset \ZZ/\ell\ZZ\)
    is a sequence of integers
    \(
        (c_0, c_1, c_2, \ldots, c_n)
    \)
    such that
    \begin{enumerate}
        \item
            every element of \(S\)
            is represented by some \(c_i \pmod{\ell}\),
        \item
            \(c_0 = 0\) and \(c_1 = 1\), 
            and
        \item
            \label{condition:DAC}
            for each \(1 < i \le n\)
            there exist \(0 \le j(i), k(i), d(i) < i\)
            such that \(c_i \equiv c_{j(i)} + c_{k(i)} \pmod{\ell}\)
            and \(c_{j(i)} - c_{k(i)} \equiv c_{d(i)} \pmod{\ell}\).
    \end{enumerate}
\end{definition}

Algorithms to enumerate \(\XSet{G}\)
using \xADD and \xDBL
correspond to 
MDACs \((c_0,\ldots,c_n)\) for \(\{1,\ldots,(\ell-1)/2\}\):
the algorithm
starts with \(x([c_0]G) = x(0) = (1:0)\)
and \(x([c_1]G) = x(G) = (X_G:Z_G)\),
then computes each \(x([c_i]G)\)
using
\[
    x([c_i]G) 
    = 
    \begin{cases}
        \xADD(x([c_{j(i)}]G),x([c_{k(i)}]G),x([c_{d(i)}]G))
        & \text{if } d(i) \not= 0
        \,,
        \\
        \xDBL([c_{j(i)}]G)
        & \text{if } d(i) = 0
        \,.
    \end{cases}
\]

\subsection{Additive kernel point enumeration}

The classic approach is to compute \(\XSet{G}\)
using repeated \xADD{}s.
Algorithm~\ref{alg:basic-enumerate}
is Costello and Hisil's \texttt{KernelPoints}
function~\cite[Algorithm~2]{costello-hisil}.
This
corresponds to the MDAC \((0,1,2,3,\ldots,(\ell-1)/2)\)
computed by repeatedly adding \(1\)
(in the notation of Definition~\ref{def:MDAC},
\((j(i),k(i),d(i)) = (i-1,1,i-2)\)),
except for \(2\) which is computed by doubling \(1\).
The simplicity of this MDAC
means that Algorithm~\ref{alg:basic-enumerate}
adapts almost trivially to \KernelRange
using a relatively small internal state:
in order to generate the next \((X_{i+1}:Z_{i+1})\),
we need only keep the values
of \((X_{i}:Z_{i})\),
\((X_{i-1}:Z_{i-1})\),
and \((X_{1}:Z_{1})\).

\begin{algorithm}
    \caption{Basic kernel point enumeration by repeated addition.
        Uses exactly \(1\) \xDBL and \((\ell-5)/2\) \xADD operations
        (for prime \(\ell > 3\)).
    }
    \label{alg:basic-enumerate}
    \DontPrintSemicolon
    \KwIn{The \(x\)-coordinate \((X_G:Z_G)\) 
        of the generator \(G\) of a cyclic subgroup 
        of order \(\ell\) in \(\EC(\FF_q)\)
    }
    \KwOut{\(\XSet{G}\) as a list}
    \((X_1:Z_1) \gets (X_G:Z_G)\)
    \;
    \((X_2:Z_2) \gets \xDBL\big((X_G:Z_G)\big)\)
    \;
    \For(\tcp*[f]{Invariant: \((X_i:Z_i) = x([i]G)\)}){\(i = 3\) \textbf{to} \((\ell-1)/2\)}{
        \(
            (X_i:Z_i) 
            \gets
            \xADD\big((X_{i-1}:Z_{i-1}),(X_{G}:Z_{G}),(X_{i-2},Z_{i-2})\big)
        \)
        \label{alg:basic-enumerate:xADD}
        \;
    }
    \Return{\(\big((X_1:Z_1),\ldots,(X_{(\ell-1)/2}:Z_{(\ell-1)/2})\big)\)}
\end{algorithm}

\subsection{\texorpdfstring{Replacing \xADD{}s with \xDBL{}s}{Replacing xADDs with xDBLs}}

Comparing \(x\)-only operations on Montgomery curves,
replacing an \(\xADD\) with an \(\xDBL\)
trades 2\Mults and 2\Adds for 1\ConstMults.
We would therefore like to replace as many \xADD{}s as possible
in our kernel enumeration with \xDBL{}s.

As a first attempt, 
we can replace 
Line~\ref{alg:basic-enumerate:xADD} of
Algorithm~\ref{alg:basic-enumerate}
with
\[
    (X_i:Z_i) 
    \gets
    \begin{cases}
        \xDBL((X_{i/2}:Z_{i/2}))
        & \text{if \(i\) is even},
        \\
        \xADD\big((X_{i-1}:Z_{i-1}),(X_{G}:Z_{G}),(X_{i-2},Z_{i-2})\big)
        & \text{if \(i\) is odd}.
    \end{cases}
\]
But applying this trick systematically
requires storing many more intermediate values,
reducing the efficiency of \KernelRange.
It also only replaces half of the \xADD{}s with \xDBL{}s,
and it turns out that we can generally do much better.

\subsection{Multiplicative kernel point enumeration}

We can do better for a large class of \(\ell\)
by considering the quotient
\[
    M_\ell := \big(\ZZ/\ell\ZZ)^\times/\subgrp{\pm1}
    \,.
\]
(We emphasize that \(M_\ell\) is a quotient of the \emph{multiplicative}
group.)
For convenience,
we write
\[
    m_\ell := \# M_\ell = (\ell-1)/2
    \,.
\]

We can now reframe the problem of enumerating \(\XSet{G}\)
as the problem of enumerating a complete set of 
representatives for \(M_\ell\).
The MDAC of Algorithm~\ref{alg:basic-enumerate}
computes the set of representatives
\(\{1,2,\ldots,m_\ell\}\),
but for the purposes of enumerating \(\XSet{G}\),
\emph{any} set of representatives will do.
Example~\ref{example:2-powering}
is particularly useful.

\begin{example}
    \label{example:2-powering}
    Suppose \(2\) generates \(M_\ell\).
    This is the case if \(2\) is a primitive element modulo \(\ell\)---that
    is, if \(2\) has order \((\ell-1)\) modulo \(\ell\)---but also if
    \(2\) has order \((\ell-1)/2\) modulo \(\ell\)
    and \(\ell \equiv 3 \pmod{4}\).
    In this case
    \[
        M_\ell = \{2^i \bmod{\ell}: 0 \le i < m_\ell\}
        \,,
    \]
    so 
    \((0,1,2,4,8,\ldots,2^{m_\ell})\)
    is an MDAC for \(M_\ell\)
    that can be computed using \emph{only}
    doubling, and \emph{no} differential additions.
\end{example}
The \KernelPoints and \KernelRange
driven by
the MDAC of Example~\ref{example:2-powering}
replace \emph{all} of the \((\ell-1)/2 - 2\) \xADD{}s
in Algorithm~\ref{alg:basic-enumerate}
with cheaper \xDBL{}s:
we save \(\ell-5\) \Mults and \(\ell-5\) \Adds at the cost of
\((\ell-5)/2\) \ConstMults.
The \KernelRange based on this MDAC
is particularly simple:
each element depends only on its predecessor,
so the internal state consists of a single
\((X_i:Z_i)\).

So, how often does this trick apply?
Theoretically,
the quantitative form of Artin's primitive root conjecture
(proven by Hooley under GRH)
says that \(M_\ell = \subgrp{2}\) for a little over half of all`\(\ell\)
(see~\cite{Wagstaff}).
Experimentally, 
\(5609420\) of the first \(10^7\) odd primes \(\ell\) satisfy \(M_\ell = \subgrp{2}\).

One might try to generalize Example~\ref{example:2-powering}
to other generators of \(M_\ell\):
for example, if \(M_\ell = \subgrp{3}\),
then we could try to find 
an MDAC for \(\{3^i\bmod\ell: 0\le i < (\ell-1)/2\}\).
But this is counterproductive:
\(x\)-only tripling (or multiplication by any scalar \(> 2\))
is \emph{slower} than differential addition.

\subsection{Stepping through cosets}

What can we do when \(M_\ell \not= \subgrp{2}\)?
A productive generalization
is to let
\[
    A_\ell := \subgrp{2} \subseteq M_\ell
    \quad
    \text{and}
    \quad
    a_\ell := \# A_\ell
    \,,
\]
and to try to compute a convenient decomposition of \(M_\ell\)
into cosets of \(A_\ell\).
Within each coset,
we can compute elements using repeated \xDBL{}s
as in Example~\ref{example:2-powering};
then, it remains to step from one coset into another using differential
additions.

This can be done in a particularly simple way for
the primes \(\ell\) such that
\begin{equation}
    \label{eq:2-3-condition}
    M_\ell \text{ is generated by } 2 \text{ and } 3.
    \tag{$\ast$}
\end{equation}
If~\eqref{eq:2-3-condition} holds, then
\[
    M_\ell = \bigsqcup_{i=0}^{m_\ell/a_\ell-1} 3^iA_\ell
    \,.
\]
We can move from the \(i\)-th to the \((i+1)\)-th coset
using the elementary relations
\begin{equation}
    \label{eq:tripling-relation}
    \begin{cases}
        c\cdot2^{j+1} + c\cdot2^j = 3c\cdot2^j
        \\
        c\cdot2^{j+1} - c\cdot2^j = c\cdot2^j
    \end{cases}
    \quad
    \text{for all integers \(c\) and \(j \ge 0\).}
\end{equation}
In particular,
if we have enumerated some coset \(3^i A_{\ell}\) by repeated doubling,
then we can compute an element of \(3^{i+1}A_{\ell}\)
by applying a differential addition to any two consecutive elements
of \(3^iA_{\ell}\) (and the difference is the first of them).
Algorithm~\ref{alg:2-3-enumerate}
minimises storage overhead by using the last two elements of the
previous coset to generate the first element of the next one.
The \KernelRange
of Algorithm~\ref{alg:2-3-enumerate}
therefore has an internal state of only two \(x\)-coordinates---so
not only is it faster than the \KernelRange of Algorithm~\ref{alg:basic-enumerate}
where it applies,
but it also has a smaller memory footprint.

\begin{algorithm}
    \caption{Kernel enumeration
        for \(\ell > 3\) satisfying~\eqref{eq:2-3-condition}.
        Cost: \((1-1/a_\ell)\cdot m_\ell\) \xDBL{}s and \(m_\ell/a_\ell-1\) \xADD{}s.
    }
    \label{alg:2-3-enumerate}
    \DontPrintSemicolon
    \KwIn{Projective \(x\)-coordinate \((X_G:Z_G)\)
    of the generator \(G\) of a cyclic subgroup of order \(\ell\) in
    \(\EC(\FF_q)\),
    where \(\ell\) satisfies~\eqref{eq:2-3-condition}.
    }
    \KwOut{\(\XSet{G}\) as a list}
    \((a,b) \gets (a_\ell,m_\ell/a_\ell)\) 
    \;
    \For(\tcp*[f]{Invariant: \((X_{ai+j}:Z_{ai+j}) = x([3^i2^{i(a-2) + (j-1)}]G)\)}){\(i = 0\) \textbf{to} \(b - 1\)}{
        \uIf{\(i = 0\)}{
            \((X_1:Z_1) \gets (X_G:Z_G)\)
            \;
        }
        \Else(\tcp*[f]{Compute new coset representative}){
            \(
                (X_{ai+1}:Z_{ai+1}) 
                \gets
                \xADD\big((X_{ai}:Z_{ai}),(X_{ai-1}:Z_{ai-1}),(X_{ai-1}:Z_{ai-1})\big)
            \)
        }
        \For(\tcp*[f]{Exhaust coset by doubling}){\(j = 2\) \textbf{to} \(a\)}{
            \(
                (X_{ai+j}:Z_{ai+j})
                \gets
                \xDBL\big((X_{ai+j-1}:Z_{ai+j-1})\big)
            \)
        }
    }
    \Return{\(\big((X_1:Z_1),\ldots,(X_{(\ell-1)/2}:Z_{(\ell-1)/2})\big)\)}
\end{algorithm}

Algorithm~\ref{alg:2-3-enumerate} performs better the closer \(a_\ell\)
is to \(m_\ell\).
In particular, when \(A_\ell = M_\ell\),
it uses \(m_\ell-1\) \xDBL{}s and no \xADD{}s at all.
The worst case for Algorithm~\ref{alg:2-3-enumerate}
is when the order of \(2\) in \(M_\ell\) is as small as possible:
that is, when \(\ell = 2^k-1\).
In this case \(a_\ell = k\),
and compared with Algorithm~\ref{alg:basic-enumerate}
we still reduce the number of \xADD{}s to be done
by a factor of \(k\).

\subsection{The remaining primes}
\label{sec:other-primes}

While 1878 of the 2261 odd primes \(\ell \le 20000\) satisfy~\eqref{eq:2-3-condition},
there are still 383 primes that do not.
We can, to some extent,
adapt Algorithm~\ref{alg:2-3-enumerate}
to handle these primes,
but on a case-by-case basis and with somewhat less satisfactory results.

For example, the CSIDH-512 parameter set
specifies 74 isogeny-degree primes 
\[
    \ell = 3, 5, 7, 11, 13, \ldots, 367, 373, \text{ and } 587
    \,.
\]
Of these 74 primes, all but seven satisfy~\eqref{eq:2-3-condition}:
the exceptions 
are \(\ell = 73\), \(97\), \(193\), \(241\), \(313\), and \(337\).
Table~\ref{tab:CSIDH-512}
lists these primes and a candidate decomposition of \(M_\ell\) into
cosets of \(A_\ell\).
In each case, 
we need to produce either an element of \(5A_\ell\) or \(7A_\ell\).
This can certainly be done using previously-computed elements,
but it requires careful tracking of those elements,
which implies a larger internal state and a more complicated execution
pattern, ultimately depending on the value of \(\ell\).

\begin{table}[ht]
    \centering
    \caption{Primes \(\ell\) in the CSIDH-512 parameter set that do not satisfy~\eqref{eq:2-3-condition}.
    }
    \label{tab:CSIDH-512}
    \begin{tabular}{r@{\;}|@{\;}r@{\;}|@{\;}c@{\;}|@{\;}l@{\;}|@{\;}l}
        \toprule
        Prime \(\ell\) & \(a_\ell\) & \([M_\ell:\subgrp{2,3}]\) & Coset decomposition of \(M_\ell\) & Notes
        \\
        \midrule
        73       & 9 & 2 & \(M_{73} = A_{73}\sqcup3A_{73}\sqcup5A_{73}\sqcup5\cdot3A_{73}\)
        % any of [ 5, 7, 10, 11, 13, 14, 15, 17, 20, 21, 22, 26, 28, 29, 30, 31, 33, 34 ]
        \\
        97       & 24 & 2 & \(M_{97} = A_{97}\sqcup 5A_{97}\) & \(3\) is in \(A_{97}\)
        % any of [ 5, 7, 10, 13, 14, 15, 17, 19, 20, 21, 23, 26, 28, 29, 30, 34, 37, 38, 39, 40, 41, 42, 45, 46 ]
        \\
        193      & 48 & 2 & \(M_{193} = A_{193}\sqcup 5A_{193}\) & \(3\) is in \(A_{193}\)
        % any of [ 5, 10, 11, 13, 15, 17, 19, 20, 22, 26, 29, 30, 33, 34, 35, 37, 38, 39, 40, 41, 44, 45, 47, 51, 52, 53, 57, 58, 60, 61, 66, 68, 70, 71, 73, 74, 76, 77, 78, 79, 80, 82, 87, 88, 89, 90, 91, 94 ]
        \\
        241      & 12 & 2 & \(M_{241} = \big(\bigsqcup_{i=0}^{4}3^iA_{241}\big) \sqcup \big(\bigsqcup_{i=0}^{4}7\cdot3^iA_{241}\big)\)
        % any of [ 7, 11, 13, 14, 17, 19, 21, 22, 23, 26, 28, 31, 33, 34, 35, 37, 38, 39, 42, 43, 44, 46, 51, 52, 55, 56, 57, 62, 63, 65, 66, 68, 69, 70, 71, 73, 74, 76, 78, 84, 85, 86, 88, 89, 92, 93, 95, 99, 101, 102, 103, 104, 105, 109, 110, 111, 112, 114, 115, 117 ]
        \\
        307      & 51 & 3 & \(M_{307} = A_{307} \sqcup 5A_{307} \sqcup 7A_{307}\) & \(3\) is in \(A_{313}\)
        % OR: use 5 and 25, or use 7 and 49
        \\
        313      & 78 & 2 & \(M_{313} = A_{313} \sqcup 5A_{313}\) & \(3\) is in \(A_{193}\)
        % any of [ 5, 7, 10, 14, 15, 17, 20, 21, 23, 28, 30, 31, 34, 37, 40, 41, 42, 43, 45, 46, 47, 51, 53, 55, 56, 59, 60, 61, 62, 63, 65, 67, 68, 69, 73, 74, 77, 80, 82, 84, 86, 89, 90, 91, 92, 93, 94, 95, 101, 102, 106, 109, 110, 111, 112, 118, 120, 122, 123, 124, 125, 126, 127, 129, 130, 131, 133, 134, 135, 136, 138, 141, 145, 146, 148, 149, 153, 154 ]
        \\
        337      & 21 & 2 & \(M_{337} =
        \big(\bigsqcup_{i=0}^33^iA_{337}\big)\sqcup\big(\bigsqcup_{i=0}^35\cdot3^iA_{337}\big)\)
        % any of [ 5, 10, 11, 15, 17, 19, 20, 22, 23, 29, 30, 31, 33, 34, 35, 38, 40, 44, 45, 46, 51, 53, 57, 58, 59, 60, 61, 62, 65, 66, 67, 68, 69, 70, 71, 73, 76, 77, 80, 83, 87, 88, 89, 90, 92, 93, 97, 99, 101, 102, 105, 106, 109, 114, 116, 118, 119, 120, 122, 124, 125, 127, 130, 132, 133, 134, 135, 136, 138, 139, 140, 142, 143, 146, 151, 152, 153, 154, 157, 159, 160, 161, 163, 166 ]
        \\
        \bottomrule
    \end{tabular}
\end{table}

\begin{example}
    Consider \(\ell = 97\).
    In this case, \(3\) is in \(A_{97}\) 
    (in fact \(3 \equiv 2^{19}\pmod{97}\)),
    and we find that \(M_{97} = A_{97}\sqcup5A_{97}\).

    To adapt Algorithm~\ref{alg:2-3-enumerate}
    to this case,
    we can still enumerate \(A_{97}\)
    using repeated doubling.
    Then,
    we need to construct an element of \(5A_{97}\) from elements of \(A_{97}\),
    for which we can use a differential addition like \(5\cdot2^i = 2^{i+2} + 2^i\)
    (with difference \(3\cdot2^i\)) or \(5\cdot2^i = 2^{i+1} + 3\cdot2^i\) (with difference \(2^i\)).
    Each involves near powers of \(2\) (modulo \(97\)),
    but also \(3\cdot2^i\)---which we know is in \(A_{97}\),
    so it need not be recomputed,
    but we need to know that \(3\cdot2^i \equiv 2^{i+19} \pmod{97}\)
    so that we can identify and store the \(x\)-coordinate corresponding
    to \(3^i\) (for a chosen \(i\)) while enumerating \(A_{97}\).
    The end result is an algorithm
    that uses one \xADD and 48 \xDBL{}s, just like
    Algorithm~\ref{alg:2-3-enumerate},
    but the internal state is slightly larger
    and the more complicated execution pattern specific to \(\ell = 97\).

    Alternatively, after (or while) enumerating \(A_{97}\),
    we could just recompute \(3 = 1 + 2\) (difference \(1\))
    to get \(5\) as \(1 + 4\) (difference \(3\))
    or \(2 + 3\) (difference \(1\)),
    but this recomputation of \(3\) is redundant.
\end{example}

Ultimately, 
there does not seem to be any ``one size fits all'' generalization of
Algorithm~\ref{alg:2-3-enumerate} for enumerating \(\XSet{G}\)
without either a more
complicated state or redundant recomputations.
The obvious approach
of finding an MDAC to enumerate a set of representatives
for \(M_\ell/\subgrp{2,3}\)
and then using Algorithm~\ref{alg:2-3-enumerate} 
to exhaust the coset containing each representative
can give reasonable results for many \(\ell\)
not satisfying~\eqref{eq:2-3-condition},
but the savings are generally not optimal.

\subsection{(In)Compatibility with Vélusqrt}

One natural question is whether these techniques can be used to further
accelerate the Vélusqrt algorithm of~\cite{BDLS20},
which can evaluate
isogenies of large prime degree \(\ell\) 
in \(\softO(\sqrt{\ell})\)
(with \(O(\sqrt{\ell})\) space).
Vélusqrt never explicitly computes all of~\(\XSet{G}\).
Instead,
it relies on the existence of a decomposition
\begin{equation}
    \label{eq:velusqrt-decomposition}
    S := \{1,3,5,\ldots,\ell-2\}
    =
    (I + J)\sqcup(I - J)\sqcup K
\end{equation}
where \(I\), \(J\), and \(K\)
are sets of integers of size \(O(\sqrt{\ell})\)
such that the maps \((i,j)\to i+j\)
and \((i,j) \to i-j\)
are injective with disjoint images.
In~\cite{BDLS20},
these sets are
\begin{align*}
    I & := \{2b(2i+1) : 0 \le i < b'\}
    \,,
    \\
    J & := \{2j + 1  0 \le j < b\}
    \,,
    \\
    K & := \{4bb' + 1, \ldots, \ell - 4, \ell - 2\}
\end{align*}
where 
\(b := \lfloor\sqrt{\ell-1}/2\rfloor\)
and
\(b' := \lfloor(\ell-1)/4b\rfloor\).
(Note that \(I\) contains ``giant steps'',
\(J\) contains ``baby steps'',
and \(K\) contains the rest of \(S\)).

The key thing to note here is that this decomposition is essentially
additive, and the elements of \(I\), \(J\), and \(K\) form arithmetic
progressions.
Algorithm~\ref{alg:2-3-enumerate},
however, is essentially multiplicative:
it works with subsets in geometric progression, not arithmetic
progression.
We cannot exclude the existence of subsets \(I\), \(J\), and \(K\)
of size \(O(\sqrt{\ell})\)
satisfying Equation~\eqref{eq:velusqrt-decomposition}
and which are amenable to enumeration by \(2\)-powering
or a variation of Algorithm~\ref{alg:2-3-enumerate}
for some, or even many \(\ell\),
but it seems difficult to construct nontrivial and useful examples.

\section{%%%%%%%%%%%%%%%%%%%%%%%%%%%%%%%%%%%%%%%%%%%%%%%%%%%%%%%%%%%%%%%%%%%%%%%
    Irrational kernel points: exploiting Frobenius
}%%%%%%%%%%%%%%%%%%%%%%%%%%%%%%%%%%%%%%%%%%%%%%%%%%%%%%%%%%%%%%%%%%%%%%%%%%%%%%%
\label{sec:frobenius}

Now suppose \(G\) is defined over a nontrivial extension \(\FF_{q^k}\)
of \(\FF_q\),
but \(\subgrp{G}\)
is defined over the subfield \(\FF_q\):
that is, it is Galois-stable.
In particular, the \(q\)-power Frobenius endomorphism \(\pi\) of \(\EC\),
which maps points in \(\EC(\FF_{q^k})\) to their conjugates
under \(\operatorname{Gal}(\FF_{q^k}/\FF_q)\),
maps \(\subgrp{G}\) into \(\subgrp{G}\). In Appendix~\ref{sec:compkernel} we 
show how we can find the point \(G\).

Since \(\pi\) maps \(\subgrp{G}\) into \(\subgrp{G}\),
it restricts to an endomorphism of \(\subgrp{G}\)---and
since the endomorphisms of \(\subgrp{G}\) are \(\ZZ/\ell\ZZ\),
and Frobenius has no kernel (so \(\pi\) is not \(0\) on \(\subgrp{G}\)),
it must act as multiplication by an eigenvalue \(\lambda \not= 0\) on \(\subgrp{G}\).
The precise value of \(\lambda\) is not important here,
but we will use the fact that 
\(\lambda\) has order \(k\) in~\((\ZZ/\ell\ZZ)^\times\)
and order \(k'\) in~\((\ZZ/\ell\ZZ)^\times/\subgrp{\pm1}\).

Now let
\[
    F_\ell := \subgrp{\lambda} \subseteq M_\ell
    \qquad
    \text{and}
    \qquad
    c_F := [M_\ell:F_\ell] = \frac{m_\ell}{k'}
    \,.
\]
Let \(R_0\) be a set of representatives for the cosets of \(F_\ell\) in
\(M_\ell\), and let \(S_0 = \{[r]G: r \in R_0\}\).
Note that
\[
    \#S_0 = (\ell-1)/k'
    \,.
\]

\subsection{The cost of Frobenius}
\label{sec:frobenius-cost}

In this section,
we seek to use the Galois action 
to replace (many) \Mults and \Squares
with a few \Frobs.
For this to be worthwhile,
\Frobs must be cheap:
and they are, even if this is not obvious
given the definition of the Frobenius map on \(\FF_{q^k}\)
as \(q\)-th powering.
It is important to note that we do not compute Frobenius
by powering.
Instead, we use the fact that Frobenius
is an \(\FF_q\)-linear map on \(\FF_{q^k}\) viewed as an
\(\FF_q\)-vector space:
that is,
Frobenius acts as a \(k\times k\) matrix 
(with entries in \(\FF_q\))
on the coefficient vectors of elements in \(\FF_{q^k}\).

The form of the Frobenius matrix, and the cost of applying it,
depends on the basis of \(\FF_{q^k}/\FF_q\).
For example:
\begin{enumerate}
    \item
        If \(k = 2\) and \(\FF_{q^2} = \FF_q(\sqrt{\Delta})\),
        then Frobenius simply negates \(\sqrt{\Delta}\)
        and the matrix is \(\operatorname{diag}(1,-1)\),
        so \(\Frobs \approx 0\).
    \item
        If \(\FF_{q^k}/\FF_q\) is represented with a normal
        basis, then the matrix represents a cyclic permutation,
        and again \(\Frobs \approx 0\).
\end{enumerate}

Even in the worst case where the basis of \(\FF_{q^k}/\FF_q\) has no
special Galois structure,
\Frobs is just the cost of multiplying a \(k\)-vector by
a \(k\times k\) matrix over \(\FF_q\):
that is, \(k^2\) \(\FF_{q}\)-multiplications
and \(k(k-1)\) \(\FF_{q}\)-additions.
This is close to the cost of a single \(\FF_{q^k}\)-muliplication
using the ``schoolbook'' method;
so when \(k \le 12\),
we have \(\Frobs \approx \Mults\) in the worst case.

\subsection{Galois orbits}

Each point $P \in \EC(\FF_{q^k})$
is contained in a \emph{Galois orbit} 
containing all the conjugates of $P$.
The kernel subgroup \(\subgrp{G}\) breaks up (as a set) into
\emph{Galois orbits}:
if we write
\[
    \mathcal{O}_P := \{ P, \pi(P), \ldots, \pi^{k-1}(P) \}
    \qquad
    \text{for } P \in \EC(\FF_{q^k})
    \,,
\]
then
\begin{equation}
    \label{eq:G-decomposition}
    \subgrp{G}
    =
    \{0\}
    \sqcup
    \begin{cases}
        \bigsqcup_{P \in S_0}\mathcal{O}_P
        & \text{if \(k\) is even},
        \\
        \big(\bigsqcup_{P \in S_0}\mathcal{O}_P\big)
        \sqcup
        \big(\bigsqcup_{P \in S_0}\mathcal{O}_{-P}\big)
        & \text{if \(k\) is odd}.
    \end{cases}
\end{equation}

To get a picture of where we are going,
recall from~\S\ref{sec:isogeny_eval_problem}
that in general,
isogeny evaluation can be reduced to evaluations
of the kernel polynomial
\[
    D(X) := \prod_{P \in S}(X - x(P))
    \,,
\]
where $S \subset \subgrp{G}$ is any subset such that
\(S \cap -S = \emptyset\) and \(S \cup -S = \subgrp{G}\setminus\{0\}\).
The decomposition of~\eqref{eq:G-decomposition}
can be seen in the factorization of \(D(X)\) over \(\FF_{q^k}\):
\begin{align*}
    D(X)
    =
    \prod_{P\in S}(X - x(P))
    & =
    \prod_{P\in S_0}\prod_{i=0}^{k'-1}(X - x(\pi^i(P)))
    \\
    & =
    \prod_{P\in S_0}\prod_{i=0}^{k'-1}(X - x(P)^{q^i})
    \,,
\end{align*}
and the factors corresponding to each \(P\) in \(S_0\)
are the irreducible factors of \(D\) over \(\FF_q\).
Transposing the order of the products,
if we let
\[
    D_0(X) := \prod_{P\in S_0}(X - x(P))
\]
then for $\alpha$ in the base field $\FF_q$, 
we can compute $D(\alpha)$ by computing \(D_0(\alpha)\)
and taking the norm:
\[
    D(\alpha) = \Norm(D_0(\alpha))
    \quad
    \text{for all }
    \alpha \in \FF_q
    \,.
\]
where
\[
    \Norm(x)
    := \prod_{i=0}^{k'-1}x^{q^i}
    =
    x(x(\cdots (x(x)^q)^q\cdots)^q)^q
    \,,
\]
which can be computed for the cost of \((k-1)\Frobs + (k-1)\Mults\)
(some multiplications can be saved with
more storage, but for small \(k\) this may not be worthwhile).

Similarly,
we can rewrite the rational map \(\phi_x\)
from~\eqref{eq:CH-phi_x-affine}
as
\[
    \phi_x(x) 
    = 
    x \cdot \Biggr(
        \prod_{P \in S} \Big( \frac{x \cdot x(P) -1 }{x - x(P)} \Big)
    \Biggl)^2
    = 
    x \cdot \Biggr(
        \prod_{P \in S_0}\prod_{i=0}^{k'-1} \Big( \frac{x \cdot x(P)^{q^i} -1 }{x - x(P)^{q^i}} \Big)
    \Biggl)^2
    \,.
\]
Evaluating \(\phi_x\) at \(\alpha\) in \(\FF_q\),
rearranging the products gives
\[
    \phi_x(\alpha) 
    = 
    \alpha \cdot \Biggr(
        \prod_{P \in S_0}\prod_{i=0}^{k'-1} \Big( \frac{\alpha \cdot
        x(P)^{q^i} -1 }{\alpha - x(P)^{q^i}} \Big)
    \Biggl)^2
    = 
    \alpha \cdot \Norm(\overline{\phi}_x(\alpha))^2
    \,,
\]
where
\[
    \overline{\phi}_x(X)
    := 
    \prod_{P \in S_0} \frac{X \cdot x(\pi^i(P)) -1 }{X - x(\pi^i(P))}
\]
Projectively,
from~\eqref{eq:CH-phi_x-projective}
we get \(\phi_x: (U:V) \mapsto (U':V')\)
where
\begin{align*}
    \label{eq:Frobenius-phi_x-projective}
    U' & = U \cdot \Big[ \prod_{i=0}^{k'-1}\prod_{P\in S_0} (U X_P^{q^i} - Z_P^{q^i} V) \Big]^2
    \,,
    \\
    V' & = V \cdot \Big[ \prod_{i=0}^{k'-1}\prod_{P\in S_0} (U Z_P^{q^i} - X_P^{q^i} V) \Big]^2
    \,,
\end{align*}
so if we set 
\begin{align*}
    F(U,V) := \prod_{P\in S_0}(U\cdot X_P - Z_P\cdot V)
    \quad\text{and}\quad
    G(U,V) := \prod_{P\in S_0}(U\cdot Z_P - X_P\cdot V)
    \,,
\end{align*}
then for \(\alpha\) and \(\beta\) in \(\FF_q\)
we get
\[
    \phi_x((\alpha:\beta))
    = (\alpha':\beta')
    :=
    \Big(
        \alpha \cdot \Norm(F(\alpha,\beta))^2
        :
        \beta \cdot \Norm(G(\alpha,\beta))^2
    \Big)
    \,.
\]

\subsection{Enumerating representatives for the Galois orbits.}
\label{sec:Galois-enumerate}

We now need to enumerate a set \(S_0\)
of representatives for the Galois orbits modulo \(\pm1\)
or, equivalently,
a set of representatives \(R_0\) for the cosets of \(F_\ell\) in~\(M_\ell\).
We therefore want good MDACs for \(M_\ell/F_\ell\).

Given an MDAC driving enumeration of the coset representatives,
there are obvious adaptations of Algorithms~\ref{alg:CH-evaluation}
and~\ref{alg:CH-evaluation-gen} to this extension field case.
Rather than iterating over all of the kernel \(x\)-coordinates,
we just iterate over a subset representing the cosets of \(\FF_\ell\),
and then compose with the norm.

Concretely, in Algorithm~\ref{alg:CH-evaluation-gen},
we should
\begin{enumerate}
    \item
        Replace \KernelRange in Line~4
        with a generator driven by an efficient
        MDAC for \(M_\ell/F_\ell\);
    \item
        Replace Line~10
        with
        \((U_i',V_i') \gets (U_i\cdot\Norm(U_i')^2,V_i\cdot\Norm(V_i')^2)\).
\end{enumerate}

First,
we can consider Algorithm~\ref{alg:basic-enumerate}:
that is, enumerating \(M_\ell/F_\ell\)
by repeated addition.
Unfortunately,
we do not have a nice bound on the length of this MDAC:
the coset representatives are not necessarily conveniently distributed
over \(M_\ell\),
so we could end up computing a lot of redundant points.

\begin{example}
    Take $(\ell,k) = (89,11)$.
    We see that \(M_\ell \not= \subgrp{\lambda,2}\).
    So we compute the minimal element in each Galois orbit (up to negation),
    and choose it to be our orbit representative.
    Using arithmetic modulo 89 only, 
    we found a choice for $R_0$ to be $R_0 = \{ 1,3,5,13 \}$.
    Now we compute an optimal MDAC, namely $(0,1,2,3,5,8,13)$.
    This chain computes the orbit generator $x$-coordinates
    using one \xDBL operation and six \xADD operations, 
    albeit involving the computation of two intermediate 
    points that will not be utilized in the final result.
\end{example}

But when we say that the coset representatives
are not conveniently distributed over \(M_\ell\),
we mean that with respect to addition.
If we look at \(M_\ell\) multiplicatively,
then the path to efficient MDACs is clearer.

The nicest case is when \(M_\ell = \subgrp{2,\lambda}\):
then,
we can take \(R_0 = \{2^i: 0 \le i < c_F\}\),
which brings us to the \(2\)-powering MDAC
of Example~\ref{example:2-powering}---except
that we stop after \(c_F-1\) \xDBL{}s.
We thus reduce the number of \xDBL{}s by a factor of \(\approx k'\),
at the expense of two norm computations.

This MDAC actually applies to more primes \(\ell\)
here than it did in~\S\ref{sec:enumerate},
because we no longer need \(2\) to generate all of \(M_\ell\);
we have \(\lambda\) to help.
In fact, the suitability of this MDAC
no longer depends on \(\ell\),
but also on \(k\).

We can go further if we assume
\begin{equation}
    \label{eq:2-3-p-condition}
    M_\ell = \subgrp{2,3,\lambda}
    \,.
    \tag{$\ast\ast$}
\end{equation}
To simplify notation, we define
\begin{align*}
    a_{\ell,k} & := [\subgrp{2,\lambda}:F_\ell]
    \,,
    &
    b_{\ell,k} & := [\subgrp{2,3,\lambda}:\subgrp{2,\lambda}]
    =
    c_F/a_{\ell,k}
    \,.
\end{align*}
Algorithm~\ref{alg:2-3-p-enumerate}
is a truncated version of Algorithm~\ref{alg:2-3-enumerate}
for computing \(S_0\) instead of \(\XSet{G}\)
when~\eqref{eq:2-3-p-condition} holds.
Algorithm~\ref{alg:2-3-p-evaluate}
is the corresponding modification of Algorithm~\ref{alg:CH-evaluation},
evaluating an \(\ell\)-isogeny over \(\FF_q\)
with kernel \(\subgrp{G}\)
at \(n\) points of \(\EC(\FF_q)\),
where \(x(G)\) is in \(\FF_{q^{k'}}\) with \(k' > 1\).

Table~\ref{tab:isogeny-cost-compare}
compares
the total cost of Algorithms~\ref{alg:2-3-p-evaluate}
and~\ref{alg:2-3-p-enumerate}
with that of Algorithms~\ref{alg:CH-evaluation} and~\ref{alg:basic-enumerate}.
In both algorithms,
we can take advantage of the fact that many of the multiplications
have one operand in the smaller field \(\FF_q\):
notably, the multiplications involving coordinates of the evaluation points.
In the context of isogeny-based cryptography (where curve constants look
like random elements of \(\FF_q\)),
this means that in Algorithm~\ref{alg:CH-evaluation},
we can replace the \(2\Mults + 2\Adds\)
in Line~\ref{alg:CH-evaluation:CrissCross}
and the \(2\Mults + 2\Squares\)
in Line~\ref{alg:CH-evaluation:square}
with \(2\ConstMults + 2\Adds\)
and \(2\ConstMults + 2\Squares\),
respectively.

\begin{table}[ht]
    \caption{Comparison of \(\ell\)-isogeny evaluation algorithms for
        kernels \(\subgrp{G}\) defined over \(\FF_q\) but with \(x(G) \in \FF_{q^{k'}}\).
        In this table, \ConstMults denotes multiplications of elements of
        \(\FF_{q^{k'}}\) by elements of \(\FF_q\) (including, but not
        limited to, curve constants).
    }
    \label{tab:isogeny-cost-compare}
    \centering
    \begin{tabular}{l|@{\;}l@{\;}|@{\;}l}
        \toprule
                    & Costello--Hisil (Algorithms~\ref{alg:CH-evaluation}
                    and~\ref{alg:basic-enumerate})
                    & This work (Algorithm~\ref{alg:2-3-p-evaluate})
        \\
        \midrule
        \Mults      & \((\ell-1)n + 2\ell - 8\) & \(2(c_F + k' - 1)n + 2c_F + 2b_{\ell,k} + 4\)
        \\
        \Squares    & \(2n + \ell - 3\)  & \(2n + 2c_F - 2\)
        \\
        \ConstMults & \((\ell+1)n + 1\)  & \(2(c_F+1)n + c_F - b_{\ell,k}\)
        \\
        \Adds       & \((n+1)(\ell+1) + 3\ell + 17\) & \(2c_F(n+1) + 4c_F + 4b_{\ell,k} - 6\)
        \\
        \Frobs      & \(0\) & \(2(k'-1)n\)
        \\
        \bottomrule
    \end{tabular}
\end{table}

\begin{algorithm}
    \caption{Compute \(S_0\) when~\eqref{eq:2-3-p-condition} holds.
    Cost: \(b_{\ell,k}-1\) \xADD{}s and \((c_F - b_{\ell,k})\) \xDBL{}s,
    or \(
        (2c_F + 2b_{\ell,k} + 4)\Mults 
        + 
        (2c_F - 2)\Squares
        +
        (c_F - b_{\ell,k})\ConstMults
        + 
        (4c_F + 4b_{\ell,k} - 6)\Adds
    \)
    }
    \label{alg:2-3-p-enumerate}
    \DontPrintSemicolon
    \KwIn{Projective \(x\)-coordinate \((X_G:Z_G)\)
        of the generator \(G\) of a cyclic subgroup of order \(\ell\) in
        \(\EC(\FF_{q^k})\),
        where \(\ell\) satisfies \(M_\ell = \subgrp{2,3,\lambda}\).
    }
    \KwOut{\(S_0\) as a list}
    \Fn{\SZeroPoints{\((X_G:Z_G)\)}}{
        \((a,b) \gets (a_{\ell,k},b_{\ell,k})\)
        \;
        \For(\tcp*[f]{Invariant: \(\XZ_{ai+j} = x([3^i2^{i(a-2) + (j-1)}]G)\)}){\(i = 0\) \textbf{to} \(b - 1\)}{
            \uIf{\(i = 0\)}{
                \(\XZ_1 \gets (X_G:Z_G)\)
                \;
            }
            \Else(\tcp*[f]{Compute new coset representative}){
                \(
                    \XZ_{ai+1}
                    \gets
                    \xADD(\XZ_{ai},\XZ_{ai-1},\XZ_{ai-1})
                \)
            }
            \For(\tcp*[f]{Exhaust coset by doubling}){\(j = 2\) \textbf{to} \(a\)}{
                \(
                    \XZ_{ai+j}
                    \gets
                    \xDBL(\XZ_{ai+j-1})
                \)
            }
        }
        \Return{\((\XZ_1,\ldots,\XZ_{c_F})\)}
    }
\end{algorithm}

\begin{algorithm}
    \caption{Isogeny evaluation using \texttt{SZeroPoints} and Frobenius.
        Cost: \(
            2(c_F + k - 1)n\Mults
            +
            2n\Squares
            +
            2(cF + 1)n\ConstMults
            +
            2c_F(n+1)\Adds
            +
            2(k-1)n\Frobs
        \)
        \emph{plus} the cost of \texttt{SZeroPoints}.
    }
    \label{alg:2-3-p-evaluate}
    \DontPrintSemicolon
    \KwIn{The \(x\)-coordinate \((X_G:Z_G)\) of a generator \(G\)
        of the kernel of an \(\ell\)-isogeny \(\phi\),
        and a list of evaluation points 
        $((U_i:V_i): 1 \le i \le n)$
    }
    \KwOut{The list of images
        \(((U_i':V_i') = \phi_x((U_i:V_i)): 1 \le i \le n)\)
    }
    \( ((X_1, Z_1),\ldots, (X_{c_F}, Z_{c_F})) \gets \)
    \SZeroPoints{\((X_G:Z_G)\)}
    \tcp*{Algorithm~\ref{alg:2-3-p-enumerate}}
    \For{\(1 \le i \le c_F\)}{
        \(
            (\hat{X_i},\hat{Z}_i)
            \gets
            (X_i + Z_i, X_i - Z_i)
        \)
        \tcp*{2\Adds} 
    }
    \For{\(i = 1\) \textbf{to} \(n\)}{
        \((\hat{U}_i,\hat{V}_i) \gets (U_i+V_i,U_i-V_i)\) 
        \tcp*{2\Adds} 
        \( (U_i',V_i') \gets (1,1) \)
        \;
        \For{\(j = 1\) \textbf{to} \(c_F\)}{
            \(
                (t_0,t_1) 
                \gets
            \)
            \CrissCross{\(\hat{X}_j\), \(\hat{Z}_j\), \(\hat{U}_i\), \(\hat{V}_i)\)}
            \tcp*{2\ConstMults + 2\Adds} % Algorithm 1 from \cite{costello-hisil}} 
            \(
                (U_i',V_i') 
                \gets
                (t_0 \cdot U_i',t_1 \cdot V_i')
            \)
            \tcp*{2\Mults} 
        }
        \((U_i',V_i') \gets (\mathtt{Norm}(U_i'), \mathtt{Norm}(V_i'))\)
        \tcp*{\(2(k'-1)\Mults + 2(k'-1)\Frobs\)}
        %% \tcp*{Algorithm~\ref{alg:norm}}
        %
        \((U_i',V_i') \gets (U_i\cdot(U_i')^2,V_i\cdot(V_i')^2)\)
        \tcp*{2\ConstMults + 2\Squares} 
    }
    \Return{\(((U'_1,V'_1),\ldots,(U'_n,V'_n))\)}
\end{algorithm}

If the parameter choice for $(\ell, k)$ does not satisfy any of these
criteria,
then we have to compute the coset representatives using some ad-hoc
MDAC.
We can do some precomputations here
to determine an optimal, or near optimal, approach to computing $R_0$.

\subsection{Experimental results}

We provide proof-of-concept implementations of our algorithms in SageMath.\footnote{Sage scripts available from \url{https://github.com/vgilchri/k-velu}.}
Our implmentations include operation-counting code
to verify the counts claimed in this article.
We provide the number of 
operations for a given $\ell$-isogeny and the extension field $k$. 
Table~\ref{tab:eval_iso} displays the costs for our algorithm,
highlighted in light gray, compared with the basic Costello--Hisil
algorithm (Algorithms~\ref{alg:CH-evaluation}
and~\ref{alg:basic-enumerate}). 
As depicted in Table~\ref{tab:eval_iso}, our approach consistently
employs fewer operations across all values of $\ell$ and extension
fields. For \(k\) of this size, it is
reasonable to use the approximation \(\Frobs \approx \Mults\)
(see~\S\ref{sec:frobenius-cost}). 

\newcommand{\greycell}{\cellcolor{gray!30}}
\begin{table}[htp]
    \caption{Cost of evaluating an \(\ell\)-isogeny at a single point
    over \(\FF_q\), using a kernel generator with \(x\)-coordinate in
    \(\FF_{q^{k'}}\).}
    \label{tab:eval_iso}
    \centering
    \begin{tabular}{r@{\ }|c|rrrrr|l}
        \toprule
        $\ell$
        & $k'$                 & \Mults  & \Squares & \ConstMults & \Adds & \Frobs & Algorithm
        \\ 
        \midrule
        \multirow{3}{*}{13}
        & any & 30 & 12 & 15 & 54 & 0 
        & Costello--Hisil (Algorithm~\ref{alg:CH-evaluation} with~\ref{alg:basic-enumerate})
        \\ 
        &   1 & \greycell 22 & \greycell 12 & \greycell 19 & \greycell 46 & \greycell 0 
        & \greycell This work (Algorithm~\ref{alg:2-3-p-evaluate})
        \\ 
        &   3 & \greycell  10 & \greycell  4 & \greycell 7 & \greycell  14 & \greycell 4
        & \greycell This work (Algorithm~\ref{alg:2-3-p-evaluate})
        \\
        \midrule
        \multirow{4}{*}{19}
        & any & 48 & 18 & 21 & 84 & 0
        & Costello--Hisil (Algorithm~\ref{alg:CH-evaluation} with~\ref{alg:basic-enumerate})
        \\
        &   1 & \greycell 34 & \greycell 18 & \greycell 28 & \greycell 70 & \greycell 0
        & \greycell This work (Algorithm~\ref{alg:2-3-p-evaluate})
        \\
        &   3 & \greycell  14 & \greycell 6 & \greycell 10 & \greycell 22 & \greycell 4
        & \greycell This work (Algorithm~\ref{alg:2-3-p-evaluate})
        \\
        &   9 &\greycell  18 & \greycell 2 & \greycell 4 & \greycell 6 & \greycell 16
        & \greycell This work (Algorithm~\ref{alg:2-3-p-evaluate})
        \\ 
        \midrule
        \multirow{3}{*}{23}  
        & any & 60 & 22 & 25 & 104 & 0 
        & Costello--Hisil (Algorithm~\ref{alg:CH-evaluation} with~\ref{alg:basic-enumerate})
        \\
        &   1 & \greycell 42 & \greycell 22 & \greycell 34 & \greycell 86 & \greycell  0
        & \greycell This work (Algorithm~\ref{alg:2-3-p-evaluate})
        \\
        &  11 & \greycell  22 & \greycell  2 & \greycell 4 & \greycell  6 & \greycell 20
        & \greycell This work (Algorithm~\ref{alg:2-3-p-evaluate})
        \\
        \bottomrule
    \end{tabular}
\end{table}

Finally,
Table~\ref{tab:nonprim-ell}
shows our success rate (over all primes \(\ell < 10^4\) )
at finding optimal MDACs for \(k \le 12\). 
These rates are computed by 
choosing the minimal representative of each Galois orbit
to be in $S_0$,
and checking whether the set $S_0$ can be computed without any intermediary additions (that would not be otherwise used).
Note, this computation checks only one approach for computing the MDAC,
hence the percentages in Table~\ref{tab:nonprim-ell} represent a lower bound
on the number of $\ell$ that have an optimal MDAC.

\begin{table}
    \caption{Percentage of primes, $3 \leq \ell < 10^4$ for which an optimal MDAC definitely exists (using the naive choice of $S_0$). }
    \label{tab:nonprim-ell}
    \centering
    \begin{tabular}{c| r| r| r| r| r| r| r| r| r| r| r | r}
        \toprule
         \textbf{k} & 1 & 2 & 3 & 4 & 5 & 6 & 7 & 8 & 9 & 10 & 11 & 12\\
         \midrule
        \textbf{\%} & 100 & 100 & 100 & 100 & 84 & 86 & 76 & 67 & 60 & 56 & 45 & 42\\
        \bottomrule
    \end{tabular}
\end{table}

\section{%%%%%%%%%%%%%%%%%%%%%%%%%%%%%%%%%%%%%%%%%%%%%%%%%%%%%%%%%%%%%%%%%%%%%%%
    Applications
}%%%%%%%%%%%%%%%%%%%%%%%%%%%%%%%%%%%%%%%%%%%%%%%%%%%%%%%%%%%%%%%%%%%%%%%%%%%%%%%
\label{sec:apps}

Our algorithms have potential applications in any isogeny-based
cryptosystem involving isogenies of prime degree \(\ell > 3\),
including key exchanges like CSIDH~\cite{csidh}
and signature schemes such as SQISign~\cite{FeoKLPW20,FeoLLW23},
SeaSign~\cite{SeaSign}, and CSI-FiSh~\cite{CSIFiSh}.
We focus on key exchange here,
but similar discussion applies for other schemes.

As mentioned in~\S\ref{sec:intro},
we also have cryptanalytic applications:
state-of-the-art attacks on group-action cryptosystems
like CSIDH involve computing a massive number of \(\ell\)-isogenies
(in order to do a Pollard-style random walk, for example,
or a baby-step giant-step algorithm as in~\cite{low-mem-csidh}).
In this context, even minor savings in individual \(\ell\)-isogenies
quickly add up to substantial overall savings.

\subsection{CSIDH and constant-time considerations}

CSIDH is a post-quantum non-interactive key exchange
based on
the action of the class group of the imaginary quadratic
order \(\ZZ[\sqrt{-p}]\)
on the set of supersingular elliptic curves \(\EC/\FF_p\)
with \(\End_{\FF_p}(\EC) \cong \ZZ[\sqrt{-p}]\).
The action is computed
via compositions of \(\ell_i\)-isogenies
for a range of small primes \((\ell_1,\ldots,\ell_m)\).

CSIDH works over prime fields \(\FF_p\),
so the methods of~\S\ref{sec:frobenius}
do not apply;
but Algorithm~\ref{alg:2-3-enumerate}
may speed up implementations
at least for the \(\ell_i\) satisfying~\eqref{eq:2-3-condition}.
(We saw in~\S\ref{sec:other-primes}
that 
67 of the 74 primes \(\ell_i\)
in the CSIDH-512 parameter set met~\eqref{eq:2-3-condition}).

The extent of any speedup depends on two factors.
The first is the number of evaluation points.
Costello and Hisil evaluate at a \(2\)-torsion point other
than \((0,0)\) in order to interpolate the image curve.
The constant-time CSIDH of~\cite{MCR19}
evaluates at one more point (from which subsequent kernels are
derived)---that is, \(n = 2\);
We find \(n = 3\) in~\cite{OAYT20},
and~\cite{CDRH-optimal}
discusses \(n > 3\).
For larger \(n\),
the cost of Algorithm~\ref{alg:CH-evaluation} overwhelms kernel enumeration,
but our results may still make a simple and interesting improvement
when \(n\) is relatively small.

The second
factor is the organisation of primes into batches for constant-time
CSIDH implementations.
CTIDH~\cite{ctidh} makes critical use of the so-called \emph{Matryoshka} property of isogeny
computations to hide the degree \(\ell\): 
using Algorithms~\ref{alg:CH-evaluation}
and~\ref{alg:basic-enumerate},
\(\ell_i\)-isogeny evaluation is a sub-computation of
\(\ell_j\)-isogeny computation whenever \(\ell_i < \ell_j\).
Organising primes into similar-sized batches,
we can add dummy operations to disguise smaller-degree isogenies
as isogenies of the largest degree in their batch.

Our Algorithm~\ref{alg:2-3-enumerate} has a limited Matryoshka property:
\(\ell_i\)-isogenies are sub-computations of \(\ell_j\)-isogenies
if \(a_{\ell_i} \le a_{\ell_k}\) 
and \(m_{\ell_i}/a_{\ell_i} \le m_{\ell_j}/a_{\ell,j}\).
For constant-time implementations,
it would make more sense to make all primes in a batch
satisfying~\eqref{eq:2-3-condition}
a sub-computation of an algorithm using the maximum \(a_\ell\) and
maximum \(m_\ell/a_\ell\) over \(\ell\) in the batch.
Redistributing batches is a delicate matter with an important impact on
efficiency;
therefore, while our work improves the running time for a fixed $\ell$, its impact on 
batched computations remains uncertain, and ultimately depends on
specific parameter choices.

\subsection{CRS key exchange}
The historical predecessors of CSIDH,
due to Couveignes~\cite{C06}
and Rostovtsev and Stolbunov~\cite{RS06,Stol09,Stol10}
are collectively known as CRS.
Here the group is the class group of a quadratic imaginary order, acting
on an isogeny (sub)class of elliptic curves with that order as their
endomorphism ring; the action is computed using a composition of
\(\ell_i\)-isogenies for a range of small primes \((\ell_1,\ldots,\ell_m)\).

In~\cite{C06,RS06,Stol09,Stol10},
isogenies are computed by finding roots of modular polynomials;
this makes key exchange extremely slow at reasonable security levels.
Performance was greatly improved in~\cite{Feo-Kief-Smi}
using Vélu-style isogeny evaluation,
but this requires
finding ordinary isogeny classes over \(\FF_p\)
with rational \(\ell_i\)-torsion points over \(\FF_{q^{k_i}}\)
with \(k_i\) as small as possible for as many \(\ell_i\) as possible.

One such isogeny class over a 512-bit prime field
is proposed in~\cite[\S4]{Feo-Kief-Smi}:
the starting curve is \(\EC/\FF_p: y^2 = x(x^2 + Ax + 1)\)
where 
\(
    p
    :=
    7\prod_\ell \ell
    - 1
\)
where the product is over all primes \(2 \le \ell <380\);
and
\[
  A =\
  \begin{subarray}{l}
    108613385046492803838599501407729470077036464083728319343246605668887327977789 \\
    32142488253565145603672591944602210571423767689240032829444439469242521864171\,.
  \end{subarray}
\]
This curve has rational \(\ell\)-isogenies 
with rational kernel generators for 
\(\ell = 3\), \(5\), \(7\), \(11\), \(13\), \(17\), \(103\), \(523\),
and \(821\),
and irrational generators over \(\FF_{q^k}\)
for \(\ell = 19\), \(29\), \(31\), \(37\), \(61\), \(71\), \(547\),
\(661\), \(881\), \(1013\), \(1181\), \(1321\), and \(1693\);
these ``irrational'' \(\ell\) are an interesting basis of comparison
for our algorithms:
all but 1321 satisfy~\eqref{eq:2-3-p-condition}.

Table~\ref{tab:crs-comparison}
compares operation counts
for Algorithms~\ref{alg:CH-evaluation} and~\ref{alg:basic-enumerate}
against Algorithm~\ref{alg:2-3-p-evaluate},
which encapsulates the improvements in~\S\ref{sec:frobenius},
for \(\ell\)-isogeny evaluation with kernel generators over
\(\FF_{q^k}\)
(and arbitrary \(n\)),
for all of the ``irrational`` \(\ell\) above \emph{except} 1321.
We see that there are substantial savings to be had for all \(n\).

\begin{table}[htp]
    \caption{Comparison of 
        Costello--Hisil (Algorithms~\ref{alg:CH-evaluation} 
        and~\ref{alg:basic-enumerate}, in white)
        with our approach (Algorithm~\ref{alg:2-3-p-evaluate}, in grey)
        for the CRS parameters with \(k > 1\) proposed in~\cite{Feo-Kief-Smi}.
        The prime \(\ell = 1321\) with \(k = 5\) is omitted,
        since in this case \(M_\ell \not= \subgrp{2,3,\lambda}\).
        In each row, \Mults, \Squares, \Adds, and \Frobs refer to
        operations on elements of \(\FF_{q^{k'}}\),
        while \ConstMults denotes multiplications of elements of
        \(\FF_{q^{k'}}\) by elements of \(\FF_q\) (including, but not
        limited to, curve constants).
    }
    \label{tab:crs-comparison}
    \centering
    \begin{tabular}{c@{\;}|@{\;}rrr@{\;}|rrrrr}
        \toprule
        \multicolumn{4}{c|}{Parameters} & \multicolumn{5}{c}{Operations}
        \\
        \(k\) & \(\ell\) & \(a_{\ell,k}\) & \(b_{\ell,k}\) & \Mults & \Squares & \ConstMults & \Adds & \Frobs
        \\
        \midrule
        \multirow{4}{*}{3} & \multirow{2}{*}{19}
        & \multirow{2}{*}{3} & \multirow{2}{*}{1}
        & 18n + 30 & 2n + 16 & 20n + 1 & 20n + 64 & 0 % CH 
        \\
        & 
        & & 
        & \greycell 10n + 4 & \greycell 2n + 4 & \greycell 8n + 2 & \greycell
8n + 14 & \greycell 4n % NEW 
        \\
        & \multirow{2}{*}{661}
        & \multirow{2}{*}{110} & \multirow{2}{*}{1}
        & 660n + 1314 & 2n + 658 & 662n + 1 & 662n + 2632 & 0 % CH 
        \\
        & 
        & & 
        & \greycell 224n + 218 & \greycell 2n + 218 & \greycell 222n + 109 & 
\greycell 222n + 656 & \greycell 4n % NEW 
        \\
        \midrule
        \multirow{4}{*}{4} & \multirow{2}{*}{1013}
        & \multirow{2}{*}{23} & \multirow{2}{*}{11}
        & 1012n + 2018 & 2n + 1010 & 1014n + 1 & 1014n + 4040 & 0 % CH 
        \\
        & 
        & & 
        & \greycell 48n + 524 & \greycell 2n + 504 & \greycell 48n + 242 & 
\greycell 48n + 1074 & \greycell 2n % NEW 
        \\
        & \multirow{2}{*}{1181}
        & \multirow{2}{*}{59} & \multirow{2}{*}{5}
        & 1180n + 2354 & 2n + 1178 & 1182n + 1 & 1182n + 4712 & 0 % CH 
        \\
        & 
        & & 
        & \greycell 120n + 596 & \greycell 2n + 588 & \greycell 120n + 290 & 
\greycell 120n + 1302 & \greycell 2n % NEW 
        \\
        \midrule
        \multirow{4}{*}{5} & \multirow{2}{*}{31}
        & \multirow{2}{*}{1} & \multirow{2}{*}{3}
        & 30n + 54 & 2n + 28 & 32n + 1 & 32n + 112 & 0 % CH 
        \\
        & 
        & & 
        & \greycell 10n + 8 & \greycell 2n + 4 & \greycell 4n & \greycell 4n
+ 14 & \greycell 8n % NEW 
        \\
        & \multirow{2}{*}{61}
        & \multirow{2}{*}{6} & \multirow{2}{*}{1}
        & 60n + 114 & 2n + 58 & 62n + 1 & 62n + 232 & 0 % CH 
        \\
        & 
        & & 
        & \greycell 20n + 10 & \greycell 2n + 10 & \greycell 14n + 5 & 
\greycell 14n + 32 & \greycell 8n % NEW 
        \\
        \midrule
        \multirow{6}{*}{7} & \multirow{2}{*}{29}
        & \multirow{2}{*}{2} & \multirow{2}{*}{1}
        & 28n + 50 & 2n + 26 & 30n + 1 & 30n + 104 & 0 % CH 
        \\
        & 
        & & 
        & \greycell 16n + 2 & \greycell 2n + 2 & \greycell 6n + 1 & \greycell
6n + 8 & \greycell 12n % NEW 
        \\
        & \multirow{2}{*}{71}
        & \multirow{2}{*}{5} & \multirow{2}{*}{1}
        & 70n + 134 & 2n + 68 & 72n + 1 & 72n + 272 & 0 % CH 
        \\
        & 
        & & 
        & \greycell 22n + 8 & \greycell 2n + 8 & \greycell 12n + 4 & 
\greycell 12n + 26 & \greycell 12n % NEW 
        \\
        & \multirow{2}{*}{547}
        & \multirow{2}{*}{39} & \multirow{2}{*}{1}
        & 546n + 1086 & 2n + 544 & 548n + 1 & 548n + 2176 & 0 % CH 
        \\
        & 
        & & 
        & \greycell 90n + 76 & \greycell 2n + 76 & \greycell 80n + 38 & 
\greycell 80n + 230 & \greycell 12n % NEW 
        \\
        \midrule
        \multirow{2}{*}{8} & \multirow{2}{*}{881}
        & \multirow{2}{*}{55} & \multirow{2}{*}{2}
        & 880n + 1754 & 2n + 878 & 882n + 1 & 882n + 3512 & 0 % CH 
        \\
        & 
        & & 
        & \greycell 116n + 220 & \greycell 2n + 218 & \greycell 112n + 108 & 
\greycell 112n + 548 & \greycell 6n % NEW 
        \\
        \midrule
        \multirow{4}{*}{9} & \multirow{2}{*}{37}
        & \multirow{2}{*}{2} & \multirow{2}{*}{1}
        & 36n + 66 & 2n + 34 & 38n + 1 & 38n + 136 & 0 % CH 
        \\
        & 
        & & 
        & \greycell 20n + 2 & \greycell 2n + 2 & \greycell 6n + 1 & \greycell
6n + 8 & \greycell 16n % NEW 
        \\
        & \multirow{2}{*}{1693}
        & \multirow{2}{*}{94} & \multirow{2}{*}{1}
        & 1692n + 3378 & 2n + 1690 & 1694n + 1 & 1694n + 6760 & 0 % CH 
        \\
        & 
        & & 
        & \greycell 204n + 186 & \greycell 2n + 186 & \greycell 190n + 93 & 
\greycell 190n + 560 & \greycell 16n % NEW 
        \\
        \bottomrule
    \end{tabular}
\end{table}

\bibliographystyle{plain}
\ifshort
\bibliography{referencesshort}
\else
\bibliography{references}
\fi
\ifsubmission
\else
\appendix
\section{%%%%%%%%%%%%%%%%%%%%%%%%%%%%%%%%%%%%%%%%%%%%%%%%%%%%%%%%%%%%%%%%%%%%%%%
    Computing a kernel generator
}%%%%%%%%%%%%%%%%%%%%%%%%%%%%%%%%%%%%%%%%%%%%%%%%%%%%%%%%%%%%%%%%%%%%%%%%%%%%%%%
\label{sec:compkernel}
One task that poses a challenge is to find a point $G \in \EC(\FF_{q^k})$. In this section, 
we will illustrate an efficient method for computing a point with the necessary properties for use in the isogeny evaluation.
\subsection{The subgroup \(H_k\).}

To compute a rational isogeny, 
our first step will be to sample a random point $P\in \EC(\FF_{q^k})$ of order $\ell$.  
For this, 
letting $N_k := \#\EC(\FF_{q^k})$, one could sample a random point 
$P$, and compute $P_\ell = [N_k/\ell]P$. 
Then $P_\ell$ is either 0 or a point of order $\ell$. 
If the order of $P_\ell$ is not $\ell$, one tries again with a new choice of $P$.

\begin{remark}\label{remark:findp}
In the special case that $\ell^2$ divides $\#\EC(\FF_{q^k})$, we instead choose 
$N_k = exp(\FF_{q^k})$, the exponent of the group order.
We do this to avoid having a cofactor, $N_k / \ell$,
that ``kills" certain torsion points. 
\end{remark}

In our context, we are assured that the $P_\ell$ we are looking for 
is \textit{not} in $\EC(\FF_q)$, 
or indeed in any \(\EC(\FF_{q^i})\), for any proper divisor \(i\) of \(k\).
We can therefore save some effort
by sampling \(P_\ell\) from the genuinely ``new'' subgroup of \(\EC(\FF_{q^k})\).

Recall that \(\EC(\FF_{q^i}) = \ker(\pi^i - [1])\) for each \(i > 0\).
For each \(k > 0\),
then
we define an endomorphism
\[
    \eta_k := \Phi_k(\pi)
    \in \End(\EC)
\]
where
\(\Phi_k(X)\) is the \(k\)-th cyclotomic polynomial
(that is, the minimal polynomial over \(\ZZ\)
of the primitive \(k\)-th roots of unity in \(\QQbar\)).
The subgroup
\[
    H_k := \ker(\eta_k) \subset \EC(\FF_{q^k})
\]
satisfies
\[
    \EC(\FF_{q^k})
    =
    H_k \oplus \sum_{i\mid k, i \not= k} \EC(\FF_{q^i})
    .
\]

The key fact is that in our situation, \(\EC[\ell](\FF_{q^k}) \subset H_k\).

\paragraph{Generating elements of \(\EC[\ell](\FF_{q^k})\).}

We always have \(\Phi_k(X) \mid X^k - 1\),
so for each \(k > 0\) 
there is an endomorphism
\[
    \delta_k := (\pi^k-[1])/\eta_k
    \in \End(\EC)
    \,,
\]
and \(\delta_k(\EC(\FF_{q^k})) \subset H_k\).
We can therefore sample a point
\(P_\ell\) in \(\EC[\ell](\FF_{q^k})\)
by computing
\[
    P_\ell = [h_k/\ell]\delta_k(P)
    \qquad
    \text{where}
    \qquad
    h_k := \#H_k
\]
and \(P\) is randomly sampled from \(\EC(\FF_{q^k})\).

\cref{tab:k-h_k-delta_k} lists 
the first few values of \(h_k\) and \(\delta_k\).
We see that evaluating \(\delta_k\)
amounts to a few Frobenius operations
(which are almost free, depending on the field representation)
and a few applications of the group law,
so this approach
saves us a factor of at least \(1/k\) in the loop length of
the scalar multiplication
(compared wih computing \(P_\ell\) as \([N_k/\ell]P\)),
but for highly composite \(k\) we save much more.

The value \(\varphi(k)\) of the Euler totient function
plays an important role.
We have \(h_k = q^{\varphi(k)} + o(q^{\varphi(k)})\),
so computing \([h_k/\ell]\) instead of \([N_k/\ell]\) 
allows us to reduce the loop length 
of basic scalar multiplication algorithms
from \(k\log_2q\) to \(\varphi(k)\log_2q\),
which is particularly advantageous when \(k\) is highly composite.

\paragraph{The action of Frobenius on \(H_k\).}

The Frobenius endomorphism \(\pi\) commutes with \(\eta_k\),
and therefore restricts to an endomorphism of \(H_k\).
If \(G \subset H_k\) is a subgroup
of prime order \(\ell\) and fixed by \(\pi\),
then \(\pi\) will act on \(G\)
as multiplication by an integer eigenvalue \(\lambda\)
(defined modulo \(\ell\)).
Since \(\eta_k = \Phi_k(\pi) = 0\) on \(H_k\) by definition,
we know that \(\lambda\) is a \(k\)-th root of unity
in \(\ZZ/\ell\ZZ\).
\paragraph{Scalar multiplication with Frobenius.}
Now, \(H_k \cong \ZZ/d_k\ZZ\times\ZZ/e_k\ZZ\),
where \(d_k\mid e_k\)
and (by the rationality of the Weil pairing)
\(d_k\mid q^k-1\).
Typically, \(d_k\) is very small compared with \(e_k\).
If \(\ell \nmid d_k\),
then we can replace \(H_k\)
with the cyclic subgroup \(H_k' := [d_k]H_k\),
and \(h_k\) with \(h_k' := e_k/d_k\).
Now, \(\pi\) induces an endomorphism of \(H_k'\),
and therefore acts as multiplication by an eigenvalue \(\lambda\)
defined modulo \(h_k'\).

We want to compute \([c_k]P\) for \(P\) in \(H_k'\),
where \(c_k := h_k'/\ell\).
Since \(\Phi_k(\pi) = 0\),
the eigenvalue \(\lambda\)
is a root of \(\Phi_k\) (i.e., a primitive \(k\)-th root of unity)
modulo \(h_k'\).
We can compute \(a_0,\ldots,a_{k-1}\)
such that
\[
    c_k \equiv \sum_{i = 0}^{k-1} a_i\lambda^i \pmod{h_k'}
    \,,
\]
with each coefficient \(a_i\approx (h_k')^{1/\varphi(k)}\) in \(O(q)\),
and then
\[
    [c_k]P = \sum_{i = 0}^{k-1} [a_i]\pi^i(P)
    \,.
\]
%\unsure{There was a HL here}
%\hl{
If we precompute the various sums of conjugates of \(P\),%}
then we can compute \([c_k]\) using a multiscalar multiplication
algorithm with a loop of length only \(\log_2q\).
This might be particularly interesting in the cases 
where \(\varphi(k) = 2\) (which corresponds to GLV multiplication) or \(4\).

\begin{table}[ht]
   \caption{The first few values of \(h_k\) and \(\delta_k\).
    }
    \label{tab:k-h_k-delta_k}
    \centering
    \begin{tabular}{r|l|l|l}
        \(k\) & \(h_k\) & \(\delta_k\) & \(\varphi(k)\)
        \\
        \hline
        \(1\)
            & \(1\)
            & \([1]\)
        \\
        \(2\)
            & \(N_2/N_1 = q + O(\sqrt{q})\)
            & \(\pi - [1]\)
            & \(1\)
        \\
        \(3\)
            & \(N_3/N_1 = q^2 + O(q^{3/2})\)
            & \(\pi - [1]\)
            & \(2\)
        \\
        \(4\)
            & \(N_4/N_2 = q^2 + O(q)\)
            & \(\pi^2 - [1]\)
            & \(2\)
        \\
        \(5\)
            & \(N_5/N_1 = q^4 + O(q^{7/2})\)
            & \(\pi - [1]\)
            & \(4\)
        \\
        \(6\)
            & \((N_6N_1)/(N_2N_3) = q^2 + O(q^{3/2})\)
            & \((\pi+[1])(\pi^3-[1])\)
            & \(2\)
        \\
        \(7\)
            & \(N_7/N_1 = q^6 + O(q^{7/2})\)
            & \(\pi - [1]\)
            & \(6\)
        \\
        \(8\)
            & \(N_8/N_4 = q^4 + O(q^2)\)
            & \(\pi^4 - [1]\)
            & \(4\)
        \\
        \(9\)
            & \(N_9/N_3 = q^6 + O(q^{4/2})\)
            & \(\pi^3 - [1]\)
            & \(6\)
        \\
        \(10\)
            & \((N_{10}N_1)/(N_2N_5) = q^4 + O(q^{7/2})\)
            & \((\pi+[1])(\pi^5 - [1])\)
            & \(4\)
        \\
        \(11\)
            & \(N_{11}/N_1 = q^{10} + O(q^{19/2})\)
            & \(\pi - [1]\)
            & \(10\)
        \\
        \(12\)
            & \((N_{12}N_2)/(N_4N_6) = q^4 + O(q^3)\)
            & \((\pi^2+[1])(\pi^6 - [1])\)
            & \(4\)
        \\
        \hline
    \end{tabular}
 
\end{table}

\begin{example}
    Consider $k=3$: we have 
    $\EC(\FF_{q^3}) \cong \EC(\mathbb{F}_{q}) \oplus H_3$,
    and $\#H_3 = N_3 / N_1$. 

    We first note that $\pi^3$ fixes the points in
    $\EC(\FF_{q^3})$, 
    so $\pi^3 - [1] = [0]$ on $\EC(\FF_{q^3})$. 
    By similar logic, the regular Frobenius map will fix the points in
    $\EC(\FF_{q})$, 
    meaning $\pi - [1] = [0]$ holds only for points contained entirely in the 
    $\EC(\FF_{q})$ portion of $\EC(\FF_{q^3})$. 
    Therefore, by computing $P_H = (\pi - 1)P$, we are ``killing" the 
    $\EC(\FF_{q})$ part of $P$, leaving only the part lying in the subgroup $H_3$.
    This computation is easy enough to do, 
    and so now we need only compute $P_\ell = [N_3/N_1/\ell]P_H$, 
    thereby saving us about a third of the multiplications. 
\end{example}

\paragraph{The ``twist trick''.}

When \(k\) is even,
if we use \(x\)-only scalar multiplication,
then the following lemma allows us 
to work over \(\FF_{q^{k/2}}\) instead of \(\FF_{q^{k}}\).
In the case \(k = 2\),
this is known as the ``twist trick''.

\begin{lemma}
    \label{lemma:twist-trick}
    If \(k\) is even,
    then every point \(P\) in \(H_k\)
    has \(x(P)\) in \(\FF_{q^{k/2}}\).
\end{lemma}

\begin{proof}
    If \(k\) is even,
    then \(\eta_k\) divides \(\pi^{k/2}+1\),
    so \(\pi^{k/2}\) acts as \(-1\) on \(H_k = \ker(\eta_k)\):
    that is, if \(P\) is in \(H_k\),
    then \(\pi^{k/2}(P) = -P\),
    so \(x(P)\) is in \(\FF_{q^{k/2}}\).
\end{proof}

\begin{example}
    Consider \(k = 6\).
    We take a random point \(R\) in \(E(\FF_{q^6})\),
    and compute \(R' := \pi^3(R) - R\),
    then \(P := \pi(R') + R'\);
    now \(P = \delta_6(R)\) is in \(H_6\),
    and \(x(P)\) is in \(\FF_{q^3}\).
    We have 
    \(h_6 = N_1^2 - (q + 1)N_1 + q^2 - q + 1 \approx q^2\),
    and we need to compute
    \(x([c_{\ell,6}]P)\) where \(c_{\ell,6} := h_6/\ell\).
    Since \(x(P)\) is in \(\FF_{q^3}\),
    we can do this using \(x\)-only arithmetic
    and the Montgomery ladder
    working entirely over \(\FF_{q^3}\).
\end{example}

The improvements outlined in this section
are summarized in 
Algorithm~\ref{alg:computeKernel}. 
    \begin{algorithm}[htp]
    \DontPrintSemicolon
    \KwIn{$\EC$ an elliptic curve defined over $\FF_{q}$, $\ell$ an integer, $k$ such that $\FF_{q^k}$ contains an $\ell$-torsion point.}
    \KwOut{$P$, a point on $\EC(\FF_{q^k})$ of order $\ell$.}
    $P_\ell \gets (0:1:0)$\;
    \Repeat{$P_\ell = (0:1:0)$}{
      $P \gets \texttt{RandomPoint}(\EC(\FF_{q^k}))$\;
      $c \gets h_k/\ell $\tcp*{$h_k$ taken from Table~\ref{tab:k-h_k-delta_k}.} 
      $P' \gets \delta_k(P)$ \tcp*{$\delta_k$ taken from Table~\ref{tab:k-h_k-delta_k}.}
      $P_\ell \gets [c]P'$
    }
     \Return $P_\ell$\;
    \caption{Computation of Kernel Generator.}
    \label{alg:computeKernel}
    \end{algorithm} 

\fi
\end{document}